\documentclass[a4paper,12pt,fleqn]{article}
\usepackage{amsmath,amssymb}
\usepackage{amsthm}
\usepackage{graphicx}
\usepackage{float}
\usepackage[np]{numprint}
\usepackage{calc}    
\usepackage{tikz,tkz-tab}
\usepackage{diagbox}
\usepackage{multirow}
\usepackage{braket}
\usepackage{varioref}
\usepackage{hyperref}
\usepackage{lipsum}
\usepackage[left=3cm,right=3cm,top=2.5cm,bottom=4cm]{geometry}
\usepackage{multirow}

\newcommand\swap{\mathtt{SWAP}}
\newcommand\cnot{\mathtt{CNOT}}
\newcommand\PauliX{Pauli-$\mathtt X$ }
\newcommand\PauliY{Pauli-$\mathtt Y$ }
\newcommand\PauliZ{Pauli-$\mathtt Z$ }
\newcommand\cx{\mathtt{CX}}

\newcommand\cz{\mathtt{CZ}}
\newcommand\XX{\mathtt{X}} 
\newcommand\YY{\mathtt{Y}} 
\newcommand\ZZ{\mathtt{Z}} 
\newcommand\UU{\mathtt{U}} 
\newcommand\II{\mathtt{I}} 
\newcommand\HH{\mathtt{H}} 
\newcommand\PP{\mathtt{P}} 

\renewcommand\phi{\varphi}
\renewcommand\epsilon{\varepsilon}
\renewcommand\geq{\geqslant}
\renewcommand\leq{\leqslant}

\newcommand\Z{\mathbb{Z}}
\newcommand\C{\mathbb{C}}
\newcommand\F{\mathbb{F}_{2}}
\newcommand\ee{\mathrm e}
\newcommand\ii{\mathrm i}
\newcommand\HS{\mathcal{H}} 
\newcommand\T{t} 
\newcommand\h{\Omega}

\newcommand\symg[1][n]{\mathfrak{S}_{#1}}  
\newcommand\PG[1][n]{\mathcal{E}_{#1}} 
\newcommand\cnotg[1][n]{{\left\langle\cnot\right\rangle}_{#1}} 
\newcommand\czg[1][n]{{\left\langle\cz\right\rangle}_{#1}} 
\newcommand\BG[1][n]{\mathcal{B}_{#1}} 
\newcommand\czxpg[1][n]{\left\langle\PP,\cz,\cnot\right\rangle_{#1}} 
\newcommand\czpg[1][n]{\left\langle\PP,\cz\right\rangle_{#1}} 
\newcommand\SL[1][n]{\mathrm{SL}_{#1}(\mathbb{F}_2)} 
\newcommand\GL[1][n]{\mathrm{GL}_{#1}(\mathbb{F}_2)} 
\newcommand\UG[1][2n]{\mathcal{U}_{2^n}} 

\newcommand\ctopzx{\mathtt{C\text{-}to\text{-}PZX}}
\newcommand\ctogpzx{\mathtt{C\text{-}to\text{-}GenPZX}}
\newcommand\BtoB{\mathtt{B\text{-}to\text{-}B_{\mathrm{red}}}}
\newcommand\AtoX{\mathtt{A\text{-}to\text{-}\cnot}}
\newcommand\red[1]{\textcolor{red}{#1}}
\newcommand\blue[1]{\textcolor{blue}{#1}}

\newcommand\vct[1]{\mathbf{#1}}
\newcommand\mat[1]{\mathbf{#1}}
\newcommand\intnf{intermediate form}
\newcommand\nf{normal form{ }}
\newcommand\pzx{PZX form{ }}
\newcommand\gpzx{GenPZX form{ }}

\newtheorem{example}{Example}
\newtheorem{theo}[example]{Theorem}
\newtheorem{prop}[example]{Proposition}

\newtheorem{lem}[example]{Lemma}

\newtheorem{defi}[example]{Definition}
\newtheorem{rem}[example]{Remark}

\begin{document}
\setlength\parindent{0mm}

\overfullrule=0mm
\floatstyle{boxed} 
\restylefloat{figure}

\title{Reduced quantum circuits for stabilizer states and graph states}

\author{Marc Bataille \\ marc.bataille1@univ-rouen.fr \\
  \\ LITIS laboratory, Universit\'e Rouen-Normandie \thanks{685 Avenue de l'Universit\'e, 76800 Saint-\'Etienne-du-Rouvray. France.}}

\date{}

\maketitle

\begin{abstract}
  We start by studying  the subgroup structures  underlying stabilizer circuits and we use our results to propose a new \nf  for stabilizer circuits. This \nf is computed by induction using simple conjugation rules in the Clifford group. It has shape  CX-CZ-P-H-CZ-P-H, where CX (resp. CZ) denotes a layer of $\cnot$ (resp. $\cz$) gates, P a layer of phase gates and H a layer of Hadamard gates. Then we consider a normal form for stabilizer states and we show how to reduce the two-qubit gate count in circuits implementing graph states. Finally we carry out a few numerical tests on classical and quantum computers in order to show the practical utility of our methods. All the algorithms described in the paper are implemented in the C language as a Linux command available on GitHub.
\end{abstract}

\section{Introduction}
In Quantum Computation, any unitary operation can be approximated to arbitrary accuracy using $\cnot$ gates together with Hadamard, Phase, and  $\pi/8$  gates (see Figure \ref{univers} for a definition of these gates and \cite[Section 4.5.3]{2011NC} for a proof of this result). Therefore, this set of gates is often called the standard set of universal gates. When we restrict this set to Hadamard, Phase and $\cnot$ gates, we obtain the set of Clifford gates. The Pauli group $\PG$ is the group generated by the Pauli gates acting on $n$ qubits (see Figure \ref{Pauli}) and the normalizer of the Pauli group in the unitary group $\UG$ is called the Clifford group.
In his PhD thesis \cite[Section 5.8]{1997G}, Gottesman gave a constructive proof of the fact that any element of the Clifford group can be expressed, up to a global phase factor, as a product of Clifford gates. He also introduced the stabilizer formalism \cite[Section 10.5.1]{2011NC}, which turned out to be is a very efficient tool to analyze quantum error-correction codes \cite{1997G} and, more generally,  to describe unitary dynamics \cite[Section 10.5.2]{2011NC}. Indeed, the Gottesman-Knill theorem asserts that a stabilizer circuit (\textit{i.e.} a quantum circuit consisting only of Clifford gates) can be simulated efficiently on a classical computer (see \cite[Section 10.5.4]{2011NC} and \cite[p. 52]{1997G}).\smallskip

In the context of quantum stabilizer circuits, the usual denomination \textit{\nf} or \textit{canonical form} just means that any stabilizer circuit is equivalent to a circuit written in this form and that  this equivalent circuit is composed of a bounded  number of Clifford gates.  Generally this equivalent circuit has the shape of a layered decomposition, each layer consisting  in a subcircuit composed of a unique type of quantum gate (\emph{e.g.} only $\cnot$ gates, only phase gates). Of course, one tries to find the shortest and simplest decomposition.  For simplicity and consistency with the previous works on this topic, we continue using the habitual expression \textit{normal form}, although a more meaningful term would be probably better suited.
Due to the importance of the Clifford gates in many fields of Quantum Computation, several normal forms for stabilizer circuits were proposed over the last two decades, with the aim of reducing the gate count in this type of circuits. Indeed, in the experimental quantum computers, the noise in the gate as well as the decoherence time are currently the main causes of their unreliability and it is  therefore imperative to minimize the number of gates in quantum circuits.
The first \nf proposed by Aaronson and Gottesman \cite{2004AG} was successively improved by Maslov and Roetteler \cite{2018MR}, Bravyi and Maslov \cite{2020BM} and
Duncan \textit{et al.} \cite{2020DKPV}. These authors used decomposition methods in the symplectic group over $\F$ in dimension $2n$ \cite{2004AG,2018MR,2020BM} or
ZX-calculus \cite{2020DKPV} in order to compute normal forms. In this paper we provide a new \nf for stabilizer circuits. This form is similar to the most recent ones \cite{2020BM,2020DKPV} but it is slightly simpler and we compute it through an original induction process based on conjugation rules in the Clifford group.\smallskip

Our result is applied to the case of stabilizer states and graph states : we propose a \nf for stabilizer states as well as a new proof of a result due to Van den Nest \textit{et al.} that asserts the local Clifford equivalence of stabilizer states and graph states \cite[theorem 1]{2004VDN}.
Graph states form an important class of stabilizer states that plays a central role in Quantum Information Theory. They are of great use in many fields such as Quantum Computing based on measurements, Quantum Error Correction, or the study of multipartite entanglement (see the numerous references given in the rich introduction of \cite{2006HD}).
We show that it is possible to reduce the two-qubit gate count in a circuit implementing a graph state by using an algorithm proposed in 2004 by Patel \textit{et al.} \cite{2004PMH} together with some conjugation rules in the Clifford group.\smallskip

This article is structured  as follows. Section \ref{background} is a background section on quantum circuits and Clifford gates that will guide the non-specialist reader through the rest of the paper. In Section \ref{groups}, we investigate some remarkable subgroups of the Clifford Group and deduce thereby a first normal form for a particular case of stabilizer circuits. In  Section \ref{NF}, we generalize this form to any stabilizer circuits.
Finally, in Section \ref{gs}, we apply this \nf to
stabilizer states and we propose an original implementation of graph states. We also provide a few simple statistics to evaluate the practical utility of our method
and we consider the case of an implementation of graph states in the publicly available IBM quantum computers.

\section{Quantum circuits and Clifford gates\label{background}}
In this background section we recall classical notions about quantum circuits and Clifford gates.  We also introduce the main notations used in the paper.
\smallskip

Let $n\geq 1$ be the number of qubit of the considered quantum register. We label each qubit from 0 to $n-1$ thus following the usual convention. For coherence we also number the lines and columns of a $n\times n$ matrix from 0 to $n-1$ and we consider that a permutation of the symmetric group $\symg$ is a bijection of $\{0,\cdots,n-1\}$. 
Bold lowercase letters 
denote a bit vector of dimension $n$, \textit{e.g.} $\vct a =[a_0,\dots,a_{n-1}]^t$, where $a_i\in\F$. In particular, the null vector of $\F^n$ is denoted by $\mathbf{0}$.  A bit matrix of dimension $n\times n$ is represented by a bold capital letter (\textit{e. g.} $\mat I$, the identity matrix, $\mat A, \mat B,\dots$). The $\oplus$ symbol denotes the addition in $\F$ (the bitwise XOR) or the symmetric difference between two sets (their union minus their intersection). The $\otimes$ symbol denotes as usual the Kronecker product of matrices or the tensor product of vector spaces. The $\odot$ symbol stands for the Hadamard product of two vectors,
\textit{i.e.} $\vct a \odot \vct b = \sum_{i=0}^{n-1}a_ib_i\vct{e}_i$, where $(\vct{e}_i)_{i=0\dots n-1} $ is the canonical basis of $\F^n$.
Unitary matrices of dimension $2^n\times 2^n$ are represented by italic capital letters (\textit{e.g.} $I$, the identity), generally labelled by one or two integers (\textit{e.g.} $X_i, Z_i, X_{ij}, Z_{ij}$), by a vector (\textit{e.g.} $X_{\vct u}, Z_{\vct v}$) or by a matrix (\textit{e.g.} $Z_{\mat B}, X_{\mat A}$). The complex number equal to $\sqrt{-1}$ is denoted by a roman $\ii$ ($\ii^2=-1$),  while the labels $i,j,k \dots$ (integers) are in italic. Classical unitary operators in dimension 2 or 4 (Figure \ref{Pauli} and \ref{univers}) are represented by typewriter uppercase letters (\textit{e.g.}
$\II, \XX, \YY, \ZZ$ in dimension 2,  $\cnot, \cz, \swap$ in dimension 4).
\smallskip

In Quantum Information Theory, a qubit is a quantum state that represents the basic information
storage unit. This state is described by a ket vector in the Dirac notation $\ket{\psi} = a_0 \ket{0} + a_1\ket{1}$ where $a_0$ and $a_1$ are complex numbers
such that $|a_0|^2 + |a_1|^2= 1$. The value of $|a_i|^2$ represents the probability that measurement produces
the value $i$. The states $\ket{0}$ and $\ket{1}$ form a basis of the Hilbert space $\HS\simeq \C^2$ where a one qubit quantum system evolves.
Operations on qubits must preserve the norm and are therefore described by unitary operators $U$ in the unitary group $\UG$.
In quantum computation, these operations are represented by quantum gates and a quantum circuit is a conventional representation of the sequence  of quantum gates applied to the qubit register over time. In Figure \ref{Pauli}, we recall the definition of the Pauli gates mentioned in the introduction. Notice that the states $\ket{0}$
and $\ket{1}$ are eigenvectors of the \PauliZ operator respectively associated to the eigenvalues 1 and -1, so the standard computational basis $(\ket{0},\ket{1})$ is also called the $\ZZ$-basis. Let $x\in\F$ be a bit. Notice that $\XX\ket{0}=\ket{1}$ and $\XX\ket{1}=\ket{0}$ (\textit{i.e.} $\XX\ket{x}=\ket{1\oplus x}$), hence the \PauliX gate is called the $\mathtt{NOT}$ gate. The phase gate $\PP$ (see Figure \ref{univers}) is defined by $\PP\ket{x}=\ii^x\ket{x}$ and the Hadamard gate $\HH$ creates superposition since $\HH\ket{x}=\frac{1}{\sqrt2}(\ket{0}+(-1)^x\ket{1})$. The following identities are used frequently in the paper. They are obtained by direct computation.
\begin{align}
  & \HH^2=\XX^2=\YY^2=\ZZ^2=\II\label{involutions}\\
  &\XX\ZZ=-\ZZ\XX\label{anticom}\\
  &\YY=\ii\XX\ZZ\label{yixz}\\
  & \HH\ZZ \HH=\XX\label{conj-z-h}\\
  &\PP^2=\ZZ\\
  &\PP\XX\PP^{-1}=\YY\label{conj-x-p}
\end{align}
The Pauli group for one qubit is the group generated by the set $\{\XX,\YY,\ZZ\}$. Any element of this group can be written uniquely in the form $\ii^{\lambda}\XX^a\ZZ^b$,
where $\lambda\in\Z_4$ and $a,b\in\F$.
\begin{figure}[h]
	\begin{center}
	\begin{tikzpicture}[scale=1.500000,x=1pt,y=1pt]
\filldraw[color=white] (0.000000, -7.500000) rectangle (24.000000, 7.500000);
\draw[color=black] (0.000000,0.000000) -- (24.000000,0.000000);
\draw[color=black] (0.000000,0.000000) node[left] {\PauliX\ \ };
\draw (55.00, 0.00) node {$\XX=\left[\begin{array}{cc}0&1\\1&0\end{array}\right]$};
\begin{scope}
\draw[fill=white] (12.000000, -0.000000) +(-45.000000:8.485281pt and 8.485281pt) -- +(45.000000:8.485281pt and 8.485281pt) -- +(135.000000:8.485281pt and 8.485281pt) -- +(225.000000:8.485281pt and 8.485281pt) -- cycle;
\clip (12.000000, -0.000000) +(-45.000000:8.485281pt and 8.485281pt) -- +(45.000000:8.485281pt and 8.485281pt) -- +(135.000000:8.485281pt and 8.485281pt) -- +(225.000000:8.485281pt and 8.485281pt) -- cycle;
\draw (12.000000, -0.000000) node {$X$};
\end{scope}

\end{tikzpicture}\qquad\qquad
\begin{tikzpicture}[scale=1.500000,x=1pt,y=1pt]
\filldraw[color=white] (0.000000, -7.500000) rectangle (24.000000, 7.500000);
\draw[color=black] (0.000000,0.000000) -- (24.000000,0.000000);
\draw[color=black] (0.000000,0.000000) node[left] {\PauliY\ \ };
\draw (55.00, 0.00) node {$\YY=\left[\begin{array}{cc}0&-i\\i&0\end{array}\right]$};
\begin{scope}
\draw[fill=white] (12.000000, -0.000000) +(-45.000000:8.485281pt and 8.485281pt) -- +(45.000000:8.485281pt and 8.485281pt) -- +(135.000000:8.485281pt and 8.485281pt) -- +(225.000000:8.485281pt and 8.485281pt) -- cycle;
\clip (12.000000, -0.000000) +(-45.000000:8.485281pt and 8.485281pt) -- +(45.000000:8.485281pt and 8.485281pt) -- +(135.000000:8.485281pt and 8.485281pt) -- +(225.000000:8.485281pt and 8.485281pt) -- cycle;
\draw (12.000000, -0.000000) node {$Y$};
\end{scope}
\end{tikzpicture}
\begin{tikzpicture}[scale=1.500000,x=1pt,y=1pt]
\filldraw[color=white] (0.000000, -7.500000) rectangle (24.000000, 7.500000);
\draw[color=black] (0.000000,0.000000) -- (24.000000,0.000000);
\draw[color=black] (0.000000,0.000000) node[left] {\PauliZ\ \ };
\draw (55.00, 0.00) node {$\ZZ=\left[\begin{array}{cc}1&0\\0&-1\end{array}\right]$};
\begin{scope}
\draw[fill=white] (12.000000, -0.000000) +(-45.000000:8.485281pt and 8.485281pt) -- +(45.000000:8.485281pt and 8.485281pt) -- +(135.000000:8.485281pt and 8.485281pt) -- +(225.000000:8.485281pt and 8.485281pt) -- cycle;
\clip (12.000000, -0.000000) +(-45.000000:8.485281pt and 8.485281pt) -- +(45.000000:8.485281pt and 8.485281pt) -- +(135.000000:8.485281pt and 8.485281pt) -- +(225.000000:8.485281pt and 8.485281pt) -- cycle;
\draw (12.000000, -0.000000) node {$Z$};
\end{scope}

\end{tikzpicture}\vspace{-4mm}

{ \caption{ The Pauli gates : names, circuit symbols and matrices \label{Pauli}}}
\end{center}
\end{figure}
\smallskip

A quantum system of two qubits $A$ and $B$ (also called  a two-qubit register) lives in a 4-dimensional Hilbert space $\HS_A\otimes\HS_B$ and the computational basis of this space is $(\ket{00}=\ket{0}_A\otimes\ket{0}_B,\ket{01}=\ket{0}_A\otimes\ket{1}_B,\ket{10}=\ket{1}_A\otimes\ket{0}_B,\ket{11}=\ket{1}_A\otimes\ket{1}_B)$.
If $\UU$ is any unitary operator acting on one qubit, a controlled-$\UU$ gate acts on the Hilbert space $\HS_{A}\otimes\HS_{B}$ as follows.
One of the two qubits (say qubit $A$) is the control qubit whereas the other qubit is the target qubit. If the control qubit $A$ is in the state $\ket 1$ then $\UU$ is applied on the target qubit $B$ but when qubit $A$ is in the state $\ket{0}$ nothing is done on qubit $B$.
The $\cnot$ gate (or $\cx$ gate) is the controlled-$\XX$ gate with control on qubit $A$ and target on qubit $B$, so the action of $\cnot$ on a two-qubit register is described by :
$\cnot\ket{00}=\ket{00}, \cnot\ket{01}=\ket{01}, \cnot\ket{10}=\ket{11}, \cnot\ket{11}=\ket{10}$ (the corresponding matrix is given in Figure \ref{univers}).
Note that this action can be sum up by the simple formula   $\cnot\ket{xy}=\ket{x,x \oplus y}$
where $\oplus$ denotes the XOR operation between two bits $x$ and $y$, which is also the addition in $\F$. In the same way, the reader can check that the controlled-$\ZZ$ operator  acts on a a basis vector as $\cz\ket{xy}=(-1)^{xy}\ket{xy}$.
Notice that this action is invariant by switching the control and the target. 
The last two-qubit gate we need is the $\swap$ gate defined by $\swap\ket{xy}=\ket{yx}$.
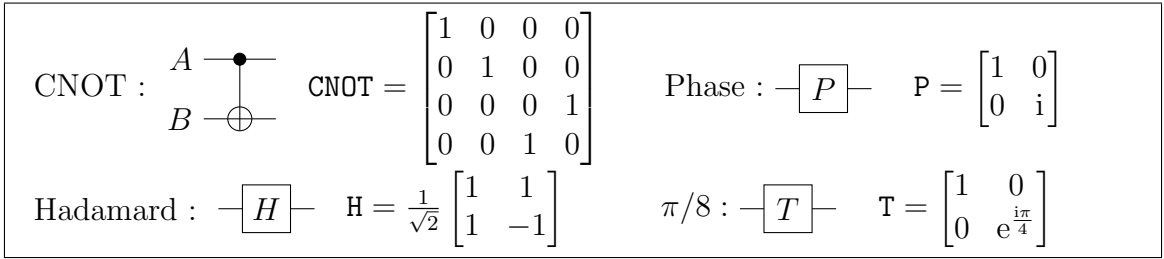
\begin{figure}[h]
\ \ CNOT :\raisebox{-7mm}{
 \begin{tikzpicture}[scale=1.500000,x=1pt,y=1pt]
\filldraw[color=white] (0.000000, -7.500000) rectangle (18.000000, 22.500000);
\draw[color=black] (0.000000,15.000000) -- (18.000000,15.000000);
\draw[color=black] (0.000000,15.000000) node[left] {$A$};
\draw[color=black] (0.000000,0.000000) -- (18.000000,0.000000);
\draw[color=black] (0.000000,0.000000) node[left] {$B$};
\draw (9.000000,15.000000) -- (9.000000,0.000000);
\begin{scope}
\draw[fill=white] (9.000000, 0.000000) circle(3.000000pt);
\clip (9.000000, 0.000000) circle(3.000000pt);
\draw (6.000000, 0.000000) -- (12.000000, 0.000000);
\draw (9.000000, -3.000000) -- (9.000000, 3.000000);
\end{scope}
\filldraw (9.000000, 15.000000) circle(1.500000pt);
\end{tikzpicture}
} 
\ $\cnot=\begin{bmatrix}
  1&0&0&0\\
  0&1&0&0\\
  0&0&0&1\\
  0&0&1&0  
\end{bmatrix}$\qquad
Phase :\raisebox{-3mm}{
\begin{tikzpicture}[scale=1.500000,x=1pt,y=1pt]
\filldraw[color=white] (0.000000, -7.500000) rectangle (24.000000, 7.500000);
\draw[color=black] (0.000000,0.000000) -- (24.000000,0.000000);
\begin{scope}
\draw[fill=white] (12.000000, -0.000000) +(-45.000000:8.485281pt and 8.485281pt) -- +(45.000000:8.485281pt and 8.485281pt) -- +(135.000000:8.485281pt and 8.485281pt) -- +(225.000000:8.485281pt and 8.485281pt) -- cycle;
\clip (12.000000, -0.000000) +(-45.000000:8.485281pt and 8.485281pt) -- +(45.000000:8.485281pt and 8.485281pt) -- +(135.000000:8.485281pt and 8.485281pt) -- +(225.000000:8.485281pt and 8.485281pt) -- cycle;
\draw (12.000000, -0.000000) node {$P$};
\end{scope}
\end{tikzpicture}
}\quad 
$\PP=\begin{bmatrix}1&0\\0&\ii\end{bmatrix}$

\raisebox{-3mm}{
\begin{tikzpicture}[scale=1.500000,x=1pt,y=1pt]
\filldraw[color=white] (0.000000, -7.500000) rectangle (24.000000, 7.500000);
\draw[color=black] (0.000000,0.000000) -- (24.000000,0.000000);
\draw[color=black] (0.000000,0.000000) node[left] {Hadamard\ :\ \ };
\begin{scope}
\draw[fill=white] (12.000000, -0.000000) +(-45.000000:8.485281pt and 8.485281pt) -- +(45.000000:8.485281pt and 8.485281pt) -- +(135.000000:8.485281pt and 8.485281pt) -- +(225.000000:8.485281pt and 8.485281pt) -- cycle;
\clip (12.000000, -0.000000) +(-45.000000:8.485281pt and 8.485281pt) -- +(45.000000:8.485281pt and 8.485281pt) -- +(135.000000:8.485281pt and 8.485281pt) -- +(225.000000:8.485281pt and 8.485281pt) -- cycle;
\draw (12.000000, -0.000000) node {$H$};
\end{scope}
\end{tikzpicture}
}\ \ 
$\HH=\frac{1}{\sqrt2}\begin{bmatrix}1&1\\1&-1\end{bmatrix}$
\hspace{10mm}
$\pi/8 : $\raisebox{-3mm}{
\begin{tikzpicture}[scale=1.500000,x=1pt,y=1pt]
\filldraw[color=white] (0.000000, -7.500000) rectangle (24.000000, 7.500000);
\draw[color=black] (0.000000,0.000000) -- (24.000000,0.000000);
\begin{scope}
\draw[fill=white] (12.000000, -0.000000) +(-45.000000:8.485281pt and 8.485281pt) -- +(45.000000:8.485281pt and 8.485281pt) -- +(135.000000:8.485281pt and 8.485281pt) -- +(225.000000:8.485281pt and 8.485281pt) -- cycle;
\clip (12.000000, -0.000000) +(-45.000000:8.485281pt and 8.485281pt) -- +(45.000000:8.485281pt and 8.485281pt) -- +(135.000000:8.485281pt and 8.485281pt) -- +(225.000000:8.485281pt and 8.485281pt) -- cycle;
\draw (12.000000, -0.000000) node {$T$};
\end{scope}
\end{tikzpicture}
}\quad
$\mathtt{T}=\begin{bmatrix}1&0\\0&\ee^{\frac{\ii\pi}{4}}\end{bmatrix}$

{ \caption{ The standard set of universal gates : names, circuit symbols and matrices. \label{univers}}}
\end{figure}
\medskip

 A $n$-qubit register evolves over time in the Hilbert space $\HS^{\otimes n}=\HS_0\otimes \HS_1\otimes\cdots\otimes \HS_{n-1}$ where $\HS_i$ is the $2$-dimensional Hilbert space of qubit $i$.
So the vector space  of an $n$-qubit system has dimension $2^n$ and
a state vector of the standard computational basis is the tensor $\ket{x_0}\otimes\dots\otimes\ket{x_{n-1}}$, where $x_i\in\{0,1\}$. This tensor  is classically denoted by
$\ket{x}$ (ket $x$), where $x$ is the binary label $x_0x_1\cdots x_{n-1}$.
Sometimes it is convenient to identify the binary label  $x=x_0x_1\cdots x_{n-1}$ of $\ket{x}$
with the column vector  $\vct x=[x_0,\dots,x_{n-1}]^{\T}=\sum_ix_i\vct{e}_i$ of the vector space $\F^n$, so one can label the vectors of the standard basis with $x$ or with $\vct x$
(\textit{i.e.} $\ket{x}=\ket{\vct x}$).
\smallskip

When we apply locally a single qubit gate $\UU$ to the qubit $i$ of a $n$-qubit register, the corresponding action on the $n$-qubit system is that of the operator
$U_i=\II\otimes\cdots\otimes\II\otimes \UU\otimes\II\otimes\cdots\otimes\II=\II^{\otimes i}\otimes \UU\otimes\II^{\otimes n-i-1}$.
As an example, if $n=4$, $ H_1=\II\otimes \HH\otimes\II\otimes\II$ and $ H_0 H_3= \HH\otimes\II\otimes \II\otimes \HH$. We also use  vectors of $\F^n$ as labels for this type of tensor. We write for example $ H_0 H_3= H_{[1,0,0,1]^{\T}}$ and more generally $H_{\vct a}=\prod_iH_i^{a_i}$. Observe that, with this notation, one has $U_i=U_{\vct{e}_i}$, $U_i^x=U_{x\vct{e}_i}$ ($x\in\{0,1\}$) and $U_{\mathbf{0}}=I$. When $\UU$ is an involution (\textit{i.e.} $\UU^2=\II$), the group generated by the $U_i$'s is isomorphic to $\F^n$, since it is an abelian 2-group. This is the case when $\UU\in\{\XX,\YY,\ZZ,\HH\}$ but not when $\UU=\PP$.
For instance
$ H_{[1,0,0,1]^{\T}} H_{[0,0,1,1]^{\T}}= H_{[1,0,0,1]^{\T}\oplus[0,0,1,1]^{\T}}= H_{[1,0,1,0]^{\T}}= H_0 H_2$. Note that the action of $ Z_i$
on $\ket{x}=\ket{x_0\cdots x_{n-1}}$ is described by
$ Z_i\ket{x}=(-1)^{x_i}\ket{x}\label{zi}$. Hence, if $\vct{v}=[v_0,\dots,v_{n-1}]^{\T}\in\F^n$, one has
\begin{equation}
 Z_{\vct v}\ket{x}=(-1)^{\vct v \cdot \vct x}\ket{x},\label{zu}
\end{equation}
where $\vct v\cdot \vct x=\sum_iv_ix_i$. In the same way, $ P_i\ket{x}=\ii^{x_i}\ket{x}$, hence
\begin{equation}
 P_{\vct v}\ket{x}=\ii^{\vct v\cdot\vct x}\ket{x}.\label{pu}
\end{equation} 
A $\cnot$ gate with target on qubit $i$ and control on qubit $j$ will be denoted $X_{ij}$ (not to be confused with $X_i$ which denotes a \PauliX gate). The reader will pay attention to the fact that our convention is the opposite of that generally used, where $\cnot_{ij}$ denotes a $\cnot$ gate with control on qubit $i$ and target on qubit $j$. The reason for this change will appear later in the proof of Theorem \ref{X-GL} (next section). So, if $i<j$, the action of $X_{ij}$ and $X_{ji}$ on a basis vector $\ket{x}$ is given by
\begin{align}
&  X_{ij}\ket{x}=X_{ij}\ket{x_0\cdots x_i\cdots x_j\cdots x_{n-1}}=\ket{x_0\cdots x_i\oplus x_j\cdots x_j\cdots x_{n-1}}\ ,\label{xij}\\
&  X_{ji}\ket{x}=X_{ji}\ket{x_0\cdots x_i\cdots x_j\cdots x_{n-1}}=\ket{x_0\cdots x_i\cdots x_j\oplus x_i\cdots x_{n-1}}.\label{xji}
\end{align}
The $\cz$  gate between qubits $i$ and $j$ is denoted by $Z_{ij}$ (not to be confused with $Z_i$ which denotes a \PauliZ gate). A $\swap$ gate between qubits $i$ and $j$ is denoted by $S_{ij}$. Notice that $Z_{ij}=Z_{ji}$ and $S_{ij}=S_{ji}$. These gates are defined by
\begin{align}
  &Z_{ij}\ket{x_0\cdots x_{n-1}}=(-1)^{x_ix_j}\ket{x}\ ,\label{zij}\\
  &S_{ij}\ket{x_0\cdots x_i\cdots x_j\cdots x_{n-1}}=\ket{x_0\cdots x_j\cdots x_i\cdots x_{n-1}}.\label{sij}
\end{align}
Observe that the $X_{ij}$, $S_{ij}$ and $ X_i$ gates are permutation matrices while the $Z_{ij}$ and $ Z_i$  gates are diagonal matrices with all diagonal entries equal to $1$ or $-1$. All these matrices are involutions. The $P_i$ gates are  also diagonal matrices but are not involutions since $ P_i^2= Z_i$.\smallskip

The three classical identities below will be of great use in the paper. They correspond to the circuit equivalences in Figure \ref{equi}. Each identity can be proved by checking that the actions of its left hand side and its right hand side on any basis vector $\ket{x}$ are the same.
\begin{align}
  X_{ij}&= H_i H_jX_{ji} H_i H_j\label{conj-xij-h}\\
  Z_{ij}&= H_iX_{ij} H_i= H_jX_{ji} H_j\label{zijxij}\\
  S_{ij}&=X_{ij}X_{ji}X_{ij}=X_{ji}X_{ij}X_{ji}\label{sijxij}
\end{align}

The Pauli group for $n$ qubits is the group generated by the set $\{ X_i, Y_i, Z_i\mid i=0\dots n-1\}$. Since
Identities \eqref{involutions}, \eqref{anticom} and \eqref{yixz} hold,  any element of this group can be written  uniquely in the form $\ii^{\lambda} X_{\vct u} Z_{\vct v}$, where $\lambda\in\Z_4$ and $\vct u,\vct v\in\F^n$. So, using \eqref{anticom}, the multiplication rule in the Pauli group is given by  
\begin{equation}\label{pauli-mult}
  \ii^{\lambda}X_{\vct u}Z_{\vct v}\ii^{\lambda'}X_{\vct{u'}}Z_{\vct{v'}}=\ii^{\lambda+\lambda'}(-1)^{\vct {u'}\cdot\vct v} X_{\vct u\oplus\vct{u'}} Z_{\vct{v}\oplus\vct{v'}}.
\end{equation}
The unitary matrix corresponding to a stabilizer circuit is an element of the group generated by the set $\{ P_i, H_i, X_{ij}\mid 0\leq i,j\leq n-1 \}$. This group contains the $S_{ij}$ and $Z_{ij}$ gates because of Identities \eqref{zijxij} and \eqref{sijxij}. It also contains the Pauli group, since $ Z_i= P_i^2$, $ X_i= H_i P_i^2 H_i$ and $ Y_i= P_i X_i P_i^{-1}= P_i H_i P_i^2 H_i P_i^3$. 
In a stabilizer circuit, changes of the overall phase by a multiple of $\frac{\pi}{4}$ are possible since
\begin{equation}
( H_i P_i)^3=( P_i H_i)^3=\ee^{\ii\frac{\pi}{4}} I.\label{hp3}
\end{equation}
This last equation can be proved by a direct computation.
\begin{figure}[h]
$\cnot$ : \ \raisebox{-5mm}{
\begin{tikzpicture}[scale=1.200000,x=1pt,y=1pt]
\filldraw[color=white] (0.000000, -7.500000) rectangle (111.000000, 22.500000);
\draw[color=black] (0.000000,15.000000) -- (111.000000,15.000000);
\draw[color=black] (0.000000,15.000000) node[left] {$A$};
\draw[color=black] (0.000000,0.000000) -- (111.000000,0.000000);
\draw[color=black] (0.000000,0.000000) node[left] {$B$};
\draw (9.000000,15.000000) -- (9.000000,0.000000);
\begin{scope}
\draw[fill=white] (9.000000, 0.000000) circle(3.000000pt);
\clip (9.000000, 0.000000) circle(3.000000pt);
\draw (6.000000, 0.000000) -- (12.000000, 0.000000);
\draw (9.000000, -3.000000) -- (9.000000, 3.000000);
\end{scope}
\filldraw (9.000000, 15.000000) circle(1.500000pt);
\draw[fill=white,color=white] (24.000000, -6.000000) rectangle (39.000000, 21.000000);
\draw (31.500000, 7.500000) node {$\sim$};
\begin{scope}
\draw[fill=white] (57.000000, 15.000000) +(-45.000000:8.485281pt and 8.485281pt) -- +(45.000000:8.485281pt and 8.485281pt) -- +(135.000000:8.485281pt and 8.485281pt) -- +(225.000000:8.485281pt and 8.485281pt) -- cycle;
\clip (57.000000, 15.000000) +(-45.000000:8.485281pt and 8.485281pt) -- +(45.000000:8.485281pt and 8.485281pt) -- +(135.000000:8.485281pt and 8.485281pt) -- +(225.000000:8.485281pt and 8.485281pt) -- cycle;
\draw (57.000000, 15.000000) node {$H$};
\end{scope}
\begin{scope}
\draw[fill=white] (57.000000, -0.000000) +(-45.000000:8.485281pt and 8.485281pt) -- +(45.000000:8.485281pt and 8.485281pt) -- +(135.000000:8.485281pt and 8.485281pt) -- +(225.000000:8.485281pt and 8.485281pt) -- cycle;
\clip (57.000000, -0.000000) +(-45.000000:8.485281pt and 8.485281pt) -- +(45.000000:8.485281pt and 8.485281pt) -- +(135.000000:8.485281pt and 8.485281pt) -- +(225.000000:8.485281pt and 8.485281pt) -- cycle;
\draw (57.000000, -0.000000) node {$H$};
\end{scope}
\draw (78.000000,15.000000) -- (78.000000,0.000000);
\begin{scope}
\draw[fill=white] (78.000000, 15.000000) circle(3.000000pt);
\clip (78.000000, 15.000000) circle(3.000000pt);
\draw (75.000000, 15.000000) -- (81.000000, 15.000000);
\draw (78.000000, 12.000000) -- (78.000000, 18.000000);
\end{scope}
\filldraw (78.000000, 0.000000) circle(1.500000pt);
\begin{scope}
\draw[fill=white] (99.000000, 15.000000) +(-45.000000:8.485281pt and 8.485281pt) -- +(45.000000:8.485281pt and 8.485281pt) -- +(135.000000:8.485281pt and 8.485281pt) -- +(225.000000:8.485281pt and 8.485281pt) -- cycle;
\clip (99.000000, 15.000000) +(-45.000000:8.485281pt and 8.485281pt) -- +(45.000000:8.485281pt and 8.485281pt) -- +(135.000000:8.485281pt and 8.485281pt) -- +(225.000000:8.485281pt and 8.485281pt) -- cycle;
\draw (99.000000, 15.000000) node {$H$};
\end{scope}
\begin{scope}
\draw[fill=white] (99.000000, -0.000000) +(-45.000000:8.485281pt and 8.485281pt) -- +(45.000000:8.485281pt and 8.485281pt) -- +(135.000000:8.485281pt and 8.485281pt) -- +(225.000000:8.485281pt and 8.485281pt) -- cycle;
\clip (99.000000, -0.000000) +(-45.000000:8.485281pt and 8.485281pt) -- +(45.000000:8.485281pt and 8.485281pt) -- +(135.000000:8.485281pt and 8.485281pt) -- +(225.000000:8.485281pt and 8.485281pt) -- cycle;
\draw (99.000000, -0.000000) node {$H$};
\end{scope}
\end{tikzpicture}
}

$\cz$ :\ \raisebox{-5mm}{
\begin{tikzpicture}[scale=1.200000,x=1pt,y=1pt]
\filldraw[color=white] (0.000000, -7.500000) rectangle (204.000000, 22.500000);
\draw[color=black] (0.000000,15.000000) -- (204.000000,15.000000);
\draw[color=black] (0.000000,15.000000) node[left] {$A$};
\draw[color=black] (0.000000,0.000000) -- (204.000000,0.000000);
\draw[color=black] (0.000000,0.000000) node[left] {$B$};
\draw (9.000000,15.000000) -- (9.000000,0.000000);
\filldraw (9.000000, 15.000000) circle(1.500000pt);
\filldraw (9.000000, 0.000000) circle(1.500000pt);
\draw[fill=white,color=white] (24.000000, -6.000000) rectangle (39.000000, 21.000000);
\draw (31.500000, 7.500000) node {$\sim$};
\begin{scope}
\draw[fill=white] (57.000000, 15.000000) +(-45.000000:8.485281pt and 8.485281pt) -- +(45.000000:8.485281pt and 8.485281pt) -- +(135.000000:8.485281pt and 8.485281pt) -- +(225.000000:8.485281pt and 8.485281pt) -- cycle;
\clip (57.000000, 15.000000) +(-45.000000:8.485281pt and 8.485281pt) -- +(45.000000:8.485281pt and 8.485281pt) -- +(135.000000:8.485281pt and 8.485281pt) -- +(225.000000:8.485281pt and 8.485281pt) -- cycle;
\draw (57.000000, 15.000000) node {$H$};
\end{scope}
\draw (78.000000,15.000000) -- (78.000000,0.000000);
\begin{scope}
\draw[fill=white] (78.000000, 15.000000) circle(3.000000pt);
\clip (78.000000, 15.000000) circle(3.000000pt);
\draw (75.000000, 15.000000) -- (81.000000, 15.000000);
\draw (78.000000, 12.000000) -- (78.000000, 18.000000);
\end{scope}
\filldraw (78.000000, 0.000000) circle(1.500000pt);
\begin{scope}
\draw[fill=white] (99.000000, 15.000000) +(-45.000000:8.485281pt and 8.485281pt) -- +(45.000000:8.485281pt and 8.485281pt) -- +(135.000000:8.485281pt and 8.485281pt) -- +(225.000000:8.485281pt and 8.485281pt) -- cycle;
\clip (99.000000, 15.000000) +(-45.000000:8.485281pt and 8.485281pt) -- +(45.000000:8.485281pt and 8.485281pt) -- +(135.000000:8.485281pt and 8.485281pt) -- +(225.000000:8.485281pt and 8.485281pt) -- cycle;
\draw (99.000000, 15.000000) node {$H$};
\end{scope}
\draw[fill=white,color=white] (117.000000, -6.000000) rectangle (132.000000, 21.000000);
\draw (124.500000, 7.500000) node {$\sim$};
\begin{scope}
\draw[fill=white] (150.000000, -0.000000) +(-45.000000:8.485281pt and 8.485281pt) -- +(45.000000:8.485281pt and 8.485281pt) -- +(135.000000:8.485281pt and 8.485281pt) -- +(225.000000:8.485281pt and 8.485281pt) -- cycle;
\clip (150.000000, -0.000000) +(-45.000000:8.485281pt and 8.485281pt) -- +(45.000000:8.485281pt and 8.485281pt) -- +(135.000000:8.485281pt and 8.485281pt) -- +(225.000000:8.485281pt and 8.485281pt) -- cycle;
\draw (150.000000, -0.000000) node {$H$};
\end{scope}
\draw (171.000000,15.000000) -- (171.000000,0.000000);
\begin{scope}
\draw[fill=white] (171.000000, 0.000000) circle(3.000000pt);
\clip (171.000000, 0.000000) circle(3.000000pt);
\draw (168.000000, 0.000000) -- (174.000000, 0.000000);
\draw (171.000000, -3.000000) -- (171.000000, 3.000000);
\end{scope}
\filldraw (171.000000, 15.000000) circle(1.500000pt);
\begin{scope}
\draw[fill=white] (192.000000, -0.000000) +(-45.000000:8.485281pt and 8.485281pt) -- +(45.000000:8.485281pt and 8.485281pt) -- +(135.000000:8.485281pt and 8.485281pt) -- +(225.000000:8.485281pt and 8.485281pt) -- cycle;
\clip (192.000000, -0.000000) +(-45.000000:8.485281pt and 8.485281pt) -- +(45.000000:8.485281pt and 8.485281pt) -- +(135.000000:8.485281pt and 8.485281pt) -- +(225.000000:8.485281pt and 8.485281pt) -- cycle;
\draw (192.000000, -0.000000) node {$H$};
\end{scope}
\end{tikzpicture}
}

$\swap$ : \ \raisebox{-5mm}{
      \begin{tikzpicture}[scale=1.200000,x=1pt,y=1pt]
\filldraw[color=white] (0.000000, -7.500000) rectangle (180.000000, 22.500000);
\draw[color=black] (0.000000,15.000000) -- (180.000000,15.000000);
\draw[color=black] (0.000000,15.000000) node[left] {$A$};
\draw[color=black] (0.000000,0.000000) -- (180.000000,0.000000);
\draw[color=black] (0.000000,0.000000) node[left] {$B$};
\draw (9.000000,15.000000) -- (9.000000,0.000000);
\begin{scope}
\draw (6.878680, 12.878680) -- (11.121320, 17.121320);
\draw (6.878680, 17.121320) -- (11.121320, 12.878680);
\end{scope}
\begin{scope}
\draw (6.878680, -2.121320) -- (11.121320, 2.121320);
\draw (6.878680, 2.121320) -- (11.121320, -2.121320);
\end{scope}
\draw[fill=white,color=white] (24.000000, -6.000000) rectangle (39.000000, 21.000000);
\draw (31.500000, 7.500000) node {$\sim$};
\draw (54.000000,15.000000) -- (54.000000,0.000000);
\begin{scope}
\draw[fill=white] (54.000000, 0.000000) circle(3.000000pt);
\clip (54.000000, 0.000000) circle(3.000000pt);
\draw (51.000000, 0.000000) -- (57.000000, 0.000000);
\draw (54.000000, -3.000000) -- (54.000000, 3.000000);
\end{scope}
\filldraw (54.000000, 15.000000) circle(1.500000pt);
\draw (72.000000,15.000000) -- (72.000000,0.000000);
\begin{scope}
\draw[fill=white] (72.000000, 15.000000) circle(3.000000pt);
\clip (72.000000, 15.000000) circle(3.000000pt);
\draw (69.000000, 15.000000) -- (75.000000, 15.000000);
\draw (72.000000, 12.000000) -- (72.000000, 18.000000);
\end{scope}
\filldraw (72.000000, 0.000000) circle(1.500000pt);
\draw (90.000000,15.000000) -- (90.000000,0.000000);
\begin{scope}
\draw[fill=white] (90.000000, 0.000000) circle(3.000000pt);
\clip (90.000000, 0.000000) circle(3.000000pt);
\draw (87.000000, 0.000000) -- (93.000000, 0.000000);
\draw (90.000000, -3.000000) -- (90.000000, 3.000000);
\end{scope}
\filldraw (90.000000, 15.000000) circle(1.500000pt);
\draw[fill=white,color=white] (105.000000, -6.000000) rectangle (120.000000, 21.000000);
\draw (112.500000, 7.500000) node {$\sim$};
\draw (135.000000,15.000000) -- (135.000000,0.000000);
\begin{scope}
\draw[fill=white] (135.000000, 15.000000) circle(3.000000pt);
\clip (135.000000, 15.000000) circle(3.000000pt);
\draw (132.000000, 15.000000) -- (138.000000, 15.000000);
\draw (135.000000, 12.000000) -- (135.000000, 18.000000);
\end{scope}
\filldraw (135.000000, 0.000000) circle(1.500000pt);
\draw (153.000000,15.000000) -- (153.000000,0.000000);
\begin{scope}
\draw[fill=white] (153.000000, 0.000000) circle(3.000000pt);
\clip (153.000000, 0.000000) circle(3.000000pt);
\draw (150.000000, 0.000000) -- (156.000000, 0.000000);
\draw (153.000000, -3.000000) -- (153.000000, 3.000000);
\end{scope}
\filldraw (153.000000, 15.000000) circle(1.500000pt);
\draw (171.000000,15.000000) -- (171.000000,0.000000);
\begin{scope}
\draw[fill=white] (171.000000, 15.000000) circle(3.000000pt);
\clip (171.000000, 15.000000) circle(3.000000pt);
\draw (168.000000, 15.000000) -- (174.000000, 15.000000);
\draw (171.000000, 12.000000) -- (171.000000, 18.000000);
\end{scope}
\filldraw (171.000000, 0.000000) circle(1.500000pt);
\end{tikzpicture}
}
{ \caption{ Classical equivalences of circuits involving $\cnot$ and Hadamard gates.\label{equi}}}
\end{figure}
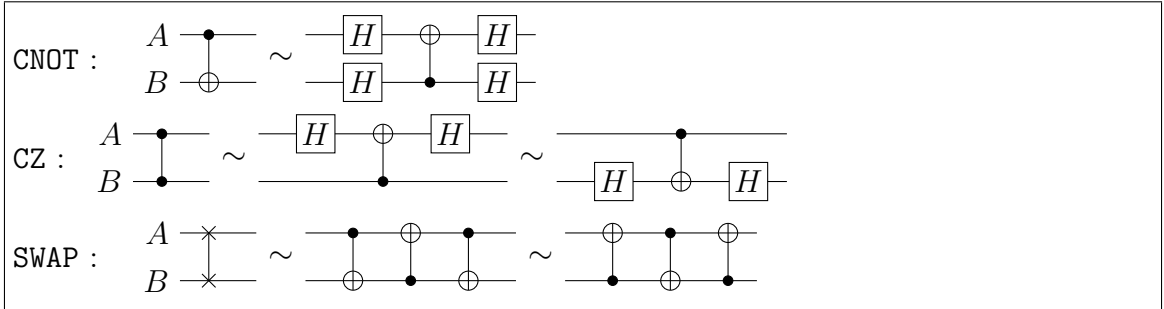

\section{Subgroup structures underlying stabilizer circuits\label{groups}}
\subsection{Quantum circuits of $\cz$ and $\cnot$ gates}
We start by describing the group $\czg$ which is the group
generated by the $Z_{ij}$ gates acting on $n$ qubits. Let us denote by $\BG$ the power set of $\{\{i,j\}\mid 0 \leq i<j\leq n-1\}$. 
As noticed in Section \ref{background}, the  matrices $Z_{ij}$ are involutions. Besides they commute with each other because they are diagonal matrices.
So $\czg$ is isomorphic to the abelian 2-group $(\BG,\oplus)$, where $\oplus$ denotes the symmetric difference of two sets. As a consequence, the order of $\czg$ is $2^{\frac{n(n-1)}{2}}$. 
For any $\mat B$ in $\BG$, we denote by $Z_{\mat B}$ the unitary operator of $\czg$ corresponding to the matrix $\mat B$, that is $Z_{\mat B}=\prod_{\{i,j\}\in \mat B}Z_{ij}$. So
the gate $Z_{ij}$ can also be denoted by $Z_{\{\{i,j\}\}}$ (we often use the notation $Z_{\{i,j\}}$ for convenience). Pay attention to the fact that $Z_{\mat B}$ denotes a product of $\cz$ gates while $Z_{\vct v}$ denotes the product of \PauliZ gates defined by the vector $\vct v$.
Using this notation, Identity \eqref{zij} can be generalized as
\begin{equation}
Z_{\mat B}\ket{x}=(-1)^{\sum_{\{i,j\}\in\mat B}x_ix_j}\ket{x}.\label{zB}
\end{equation}
To any $\mat B$ in $\BG$, we associate a  $\F$ matrix of dimension $n\times n$,  whose entry $(i,j)$
is 1 when $\{i,j\}$ is in $\mat B$ and 0 otherwise. These matrices are symmetric with only zeros on the diagonal and they form an additive group isomorphic to $(\BG,\oplus)$. So, in practice, one can identify the elements of $\BG$ to matrices. For example $\{\{i,j\}\}$ also denotes the matrix whose entries are all 0
but entries $(i,j)$ and $(j,i)$ that are equal to 1.
Let $q_{\mat B}$ be the quadratic form defined on $\F^n$ by
\begin{equation}
  q_{\mat B}(\vct x)=\sum_{\{i,j\}\in \mat B}x_ix_j=\sum_{i<j}b_{ij}x_ix_j,
\end{equation}
where $b_{ij}$ is the entry $(i,j)$ of matrix $\mat B$. Then Identity \eqref{zB} can be rewritten as
\begin{equation}
  Z_{\mat B}\ket{x}=(-1)^{q_{\mat B}(\vct x)}\ket{x}.
\end{equation}
Note that $\mat B$ can be viewed as the matrix of the alternating (and symmetric) bilinear form associated to the quadratic form $q_{\mat B}$.\medskip

In a previous work \cite{2020B}, we described  the group $\cnotg$ generated by the $X_{ij}$ gates acting on $n$ qubits. We recall now some results from this work. The special linear group on any field $K$ is generated by the set of transvection matrices. In the special case of $K=\F$, this set is reduced to the $n(n-1)$ matrices $\mat I \oplus \mat{E}_{ij}$, where
$\mat{E}_{ij}$ is the matrix with all entries 0 except the entry $(i,j)$ that is equal to 1.  Let us denote by $[ij]$
the transvection matrix $\mat{I} \oplus \mat{E}_{ij}$. The general linear group $\GL$ is equal to $\SL$, the special linear group on $\F$, and is consequently generated by the matrices $[ij]$. The following simple property of the matrices $[ij]$ will be frequently used in the rest of the article.
    \begin{prop}\label{tij-mult}
Multiplying to the left (resp. the right) any matrix $M$ by a 
transvection matrix $[ij]$ is equivalent to add the row $j$ (resp. column $i$) to the row $i$ (resp. column $j$) in $M$.
\end{prop}
Applying Proposition \ref{tij-mult} to the column vector $\mat x\in\F^n$ corresponding to the binary label $x$ of the basis vector $\ket{x}$, we can rewrite Relation \eqref{xij} in a cleaner way as
\begin{equation}
 X_{ij}\ket{x}=\ket{[ij]\mat x}.\label{xijtij}
 \end{equation}
 The above considerations lead quite naturally to the following theorem.

\begin{theo}\label{X-GL}
       
      The group $\cnotg$ generated by the $\cnot$ gates acting on $n$ qubits is isomorphic to $\GL$.
      The morphism $\Phi$ sending each gate $X_{ij}$ to the transvection matrix $[ij]$ is an explicit isomorphism.
      The order of $\cnotg$ is $2^{\frac{n(n-1)}{2}}\prod_{i=1}^n(2^i-1)$.
    \end{theo}

    \begin{proof}
      As the matrices $[ij]$ generate $\GL$, it is clear that  $\Phi$ is surjective. Since Identity \eqref{xijtij} holds, a preimage $U$ under $\Phi$ of any matrix $\mat A$ in $\GL$ must satisfy  the relations $U\ket{x}=\ket{\mat{Ax}}$ for any basis vector $\ket{x}$. As these relations define a unique matrix  $U$, $\Phi$ is injective. The order of $\GL$ is classically obtained by counting the number of basis of the vector space $\F^n$.
\end{proof}

For any $\mat A$ in $\GL$, let $X_{\mat A}=\Phi^{-1}(\mat A)$, where $\Phi$ is the morphism defined in Theorem \ref{X-GL}. The unitary operator $X_{\mat A}$ thus corresponds to any circuit composed of the $\cnot$ gates $X_{i_1 j_1}\dots X_{i_\ell j_\ell}$ such that $\mat A=\prod_{k=1}^{\ell}[i_kj_k]$ and the gate $X_{ij}$ can also be denoted by $X_{[ij]}$. Pay attention to the fact that $X_{\mat A}$ denotes a product of $\cnot$ gates while $X_{\vct u}$ denotes the product of \PauliX gates defined by the vector $\vct u$. As $([ij][jk])^2=[ik]$, a straightforward consequence of the isomorphism between $\cnotg$ and $\GL$ is the following conjugation rule between the $\cnot$ gates.
\begin{equation}
  X_{[ij]}X_{[jk]}X_{[ij]}=X_{[ik]}X_{[jk]}=X_{[jk]}X_{[ik]}\quad (i,j,k \text{ distinct})\label{conj-xij-xjk}
\end{equation}

\subsection{The \pzx for quantum circuits of phase, $\cz$ and $\cnot$ gates}
Let $\czpg$ be the group generated by the set $\{ P_i,Z_{ij}\mid 0\leq i,j\leq n-1\}$. 
Any element of the group generated by the $ P_i$ gates can be written uniquely in the form $ Z_{\vct v} P_{\vct b}$ where $\vct v,\vct b\in\F^n$.
This group is isomorphic to $(\Z_4^n,+)$, one possible isomorphism associating $ Z_{\vct v} P_{\vct b}$ to $2\vct v+ \vct b$.
As the generators of the group $\czpg$ commute between each other, the group $\czpg$ is isomorphic to the direct product $\Z_4^n\times\BG$. Any element in $\czpg$ can be written uniquely in the form
$ Z_{\vct v} P_{\vct b}Z_{\mat B}$ and 

\begin{equation}
   Z_{\vct{v}} P_{\vct{b}}Z_{\mat B} Z_{\vct{v'}} P_{\vct{b'}}Z_{\mat{B'}}= Z_{\vct{v}\oplus\vct{v'}\oplus\vct{b}\odot\vct{b'}} P_{\vct{b}\oplus\vct{b'}}Z_{\mat B\oplus \mat{B'}}.\label{czpg-mult}
\end{equation}
The conjugation by the $X_{[ij]}$ gates in $\czpg$ obey to the seven rules below.
Each equality can be proved by checking, thanks to Identities \eqref{zu} to \eqref{zij}, that the actions of its left hand side and its right hand side on any basis vector $\ket{x}$ are the same.
\begin{align}
  &X_{[ij]}Z_{\{i,j\}}X_{[ij]} = Z_{\{i,j\}} Z_j\label{conj-Zij}\\
  &X_{[ij]}Z_{\{i,k\}}X_{[ij]} = Z_{\{i,k\}}Z_{\{j,k\}}\quad (i,j,k\ \text{distinct})\label{conj-Zik}\\
  &X_{[ij]}Z_{\{p,q\}}X_{[ij]} = Z_{\{p,q\}}\quad (p,q\neq i)\label{conj-Zpq}\\
  &X_{[ij]}  Z_{i}X_{[ij]}=  Z_i Z_j\label{conj-Zi}\\
  &X_{[ij]}  Z_{j}X_{[ij]}=  Z_j\label{conj-Zj}\\
  &X_{[ij]}  P_{i}X_{[ij]}=  P_i P_jZ_{\{i,j\}}\label{conj-Pi}\\
  &X_{[ij]}  P_{j}X_{[ij]}=  P_j\label{conj-Pj}
\end{align}

Let us denote by $\czxpg$ the group generated by the set $\{ P_i,Z_{ij},X_{ij}\mid 0\leq i,j\leq n-1\}$. As described in the following proposition, we can extend relations \eqref{conj-Zij} to \eqref{conj-Pj} to the unitary matrices $ Z_{\vct{v}}, P_{\vct{b}}$ and $Z_{\mat{B}}$.

\begin{prop}\label{normal}
  The group $\czpg$ is a normal subgroup of $\czxpg$. The conjugation of any element of $\czpg$ by a $\cnot$ gate is described by the relations
\begin{align}
  & X_{[ij]}  Z_{\vct v}X_{[ij]}=  Z_{[ji]\vct v}\ ,\label{conj-Za-xij}\\
  &X_{[ij]}  P_{\vct b}X_{[ij]}=Z_{b_i b_j \vct{e}_j} P_{[ji]\vct b}Z_{\{i,j\}}^{b_i}= Z_{b_i b_j \vct{e}_j} P_{[ji]\vct b}Z_{b_i\{\{i,j\}\}}\ ,\label{conj-Pb-xij}\\
  &X_{[ij]}Z_{\mat B}X_{[ij]} =  Z_{b_{ij}\vct{e}_j}Z_{[ji] \mat{B} [ij]}\ ,\label{conj-ZB-xij}
 \end{align}
\end{prop}

\begin{proof}
  Identities \eqref{conj-Za-xij} and \eqref{conj-Pb-xij} are direct consequences of the conjugation relations \eqref{conj-Zi}, \eqref{conj-Zj}, \eqref{conj-Pi}, \eqref{conj-Pj} and Proposition \eqref{tij-mult} applied to the vectors $\vct v$ and $\vct b$.
  Let us prove Identity \eqref{conj-ZB-xij}.
  Let $\mat{B}_i=\{\{p,q\}\in \mat B\mid i\in\{p,q\}\}$, $\mat{B}_i^{c}=\mat{B}_i\oplus \mat B$ and $\mat{B}_i'=\mat{B}_i\oplus b_{ij}\{\{i,j\}\}$, then $\mat B=b_{ij}\{\{i,j\}\}\oplus \mat{B}_i'\oplus \mat{B}_i^c$.
  On one hand,   $[ji]\mat B[ij]=b_{ij}[ji]\{\{i, j\}\}[ij]\oplus[ji]\mat{B}_i'[ij]\oplus [ji]\mat{B}_i^c[ij]$.
  We check that $[ji]\{\{i, j\}\}[ij] = \{\{i, j\}\}$,
  $[ji]\{\{i, k\}\}[ij] = \{\{i, k\}, \{j, k\}\}$ when $k\neq j$ and $[ji]\{\{p, q\}\}[ij] = \{\{p, q\}\}$ when $p,q\neq i$ 
  (recall that $\{\{p,q\}\}$ denotes the matrix in $\BG$ whose entries are 0 but entries $(p,q)$ and $(q,p)$ that are 1).
  Hence
  \begin{equation}\label{one-hand}
    Z_{[ji]\mat B[ij]}=Z_{ij}^{b_{ij}}Z_{\mat{B}_i^c}\prod_{k\in \Lambda_{i}}Z_{ik}Z_{jk},
    \end{equation}
  where $\Lambda_{i}=\{k \mid \{i,k\}\in \mat{B}_i'\}$.
On the other hand, $X_{[ij]}Z_{\mat B}X_{[ij]}=X_{[ij]}Z_{ij}^{b_{ij}}Z_{\mat{B}_i'}Z_{\mat{B}_i^c}X_{[ij]}$, so using \eqref{conj-Zij}, \eqref{conj-Zik}
  and \eqref{conj-Zpq}, one has
  \begin{equation}\label{other-hand}
    X_{[ij]}Z_{\mat B}X_{[ij]}=Z_{ij}^{b_{ij}} Z_j^{b_{ij}}Z_{\mat{B}_i^c}\prod_{k\in \Lambda_i}Z_{ik}Z_{jk}.
  \end{equation}
  As $ Z_j^{b_{ij}}= Z_{b_{ij}\mat{e}_j}$, we conclude by comparing \eqref{one-hand} and \eqref{other-hand}.
  \end{proof}

  We can extend Identity \ref{conj-ZB-xij} to the case of any unitary matrix $X_{\mat A}$.
  
  \begin{prop}\label{ZB-XA-prop}
    For any matrix $\mat B$ in $\BG$ and any matrix $\mat A$ in $\GL$, one has 
 \begin{equation}
  X_{\mat{A}}Z_{\mat{B}}X_{\mat{A}}^{-1}= Z_{q_{\mat B}(\mat A^{-1})}Z_{\mat{A}^{-\T}\mat{B}\mat{A}^{-1}}\ ,\label{conj-ZB-XA}
 \end{equation}
 where $q_{\mat B}$ is the quadratic form defined by $\mat B$,
 $q_{\mat B}(\mat A)$ is a shorthand for the vector $[q_{\mat B}(\vct{c_0}),\dots,q_{\mat B}(\vct{c_{n-1}})]^{\T}$, 
 $\vct{c}_0,\dots, \vct{c}_{n-1}$ are the columns of matrix $\mat A$
 and $\mat{A}^{-\T}$ is a shorthand for $\left(\mat{A}^{\T}\right)^{-1}$.
\end{prop}

\begin{proof}
    Since Identities \eqref{conj-ZB-xij} and \eqref{conj-Za-xij} hold, it is clear that $X_{\mat{A}}Z_{\mat{B}}X_{\mat{A}}^{-1}$ can be written in the form
  $ Z_{\vct v}Z_{\mat A^{-\T}\mat B \mat A^{-1}}$ for some $\vct v$ in $\F^n$. So we have to prove that $\vct v=q_{\mat B}(\mat A^{-1})$. We start from $ Z_{\mat v}=X_{\mat{A}}Z_{\mat{B}}X_{\mat{A}}^{-1}Z_{\mat{B'}}$, where
  $\mat{B'}=\mat A^{-\T}\mat B \mat A^{-1}$. Let $\ket{\psi}= Z_{\vct v}\ket{\vct{e_i}}$, then $\ket{\psi}=(-1)^{v_i}\ket{\vct{e_i}}$.
  On the other hand, $\ket{\psi}=X_{\mat{A}}Z_{\mat{B}}X_{\mat{A}}^{-1}\ket{\vct{e_i}}$ since $q_{\mat{B'}}(\vct{e_i})=0$ for any $\vct{B'}\in\BG$. Besides, $X_{\mat{A}}^{-1}\ket{\vct{e_i}}= \ket{\mat A^{-1}\vct{e_i}}=\ket{\vct{c_i}}$ where $\vct{c_i}$ is the  column $i$ of $\mat A^{-1}$, hence
  $\ket{\psi}=X_{\mat{A}}Z_{\mat{B}}\ket{\vct{c}_i}=(-1)^{q_{\mat B}(\vct{c}_i)}X_{\mat{A}}\ket{\vct{c}_i}=(-1)^{q_{\mat B}(\vct{c}_i)}\ket{\vct{e}_i}$.
  Finally we see that $v_i= q_{\mat B}(\vct{c}_i)$, thus $\vct v=q_{\mat B}(\mat A^{-1})$.
  \end{proof} 
  
  From Identity \eqref{sijxij}, the gate $S_{ij}$ is in $\cnotg$ and is therefore a  gate of type $X_{\mat{A}}$. Let $(ij)$ be the permutation matrix of $\GL$ associated to the transposition $\tau$ of $\symg$ that swaps  $i$ and $j$, then $(ij)=[ij][ji][ij]=[ji][ij][ji]$, hence
  $S_{ij}=X_{[ij]}X_{[ji]}X_{[ij]}=X_{[ij][ji][ij]}=X_{(ij)}$.
  The group generated by the $X_{(ij)}$ gates is a subgroup of $\cnotg$ that is isomorphic to $\symg$ and we denote by $X_{\mat \sigma}$ the unitary matrix associated to the permutation matrix $\mat \sigma$ in $\GL$. The conjugation by $X_{\sigma}$ is given by $X_{\sigma}Z_{\{p,q\}}X_{\sigma}=Z_{\{\sigma(p),\sigma(q)\}}$ and, in particular, one has  $X_{(ij)}Z_{\{p,q\}}X_{(ij)}=Z_{\{\tau(p),\tau(q)\}}$ (see \cite{2019BL} for further development on $\cz$ and $\swap$ gates). As a consequence of Propositions \ref{normal} and \ref{ZB-XA-prop}, the followings identities hold :
\begin{align}
  & X_{(ij)}  Z_{\vct v}X_{(ij)}=  Z_{(ij)\vct v}\ , \label{conj-Za-sij}\\
  &X_{(ij)}  P_{\vct b}X_{(ij)}=  P_{(ij)\vct b}\ ,\label{conj-Pb-sij}\\
  &X_{(ij)}Z_{\mat B}X_{(ij)} = Z_{(ij)\mat{B}(ij)}\ ,\label{conj-ZB-sij}\\
  &X_{\sigma}Z_{\mat{B}}X_{\sigma}^{-1}= Z_{\sigma \mat B\sigma^{-1}} \text{\ (for any permutation matrix } \sigma ). \label{conj-ZB-sigma}
\end{align}

Proposition \ref{normal} provides straightforwardly an algorithm to write in \nf any quantum circuit $C$ composed of
$\PP,\cz$ and $\cnot$ gates. This \nf is called the \pzx\ (Theorem \ref{decompo})  and the algorithm is called the $\ctopzx$ algorithm (Figure \ref{C-to-PZX}).

\begin{figure}[h]
  $\mathtt{ALGORITHM}$ : Compute the \pzx for a stabilizer circuit of $\PP$, $\cz$ and $\cnot$ gates.
  
    $\mathtt{INPUT} :  (C,F_{\text{in}})$, where
    
    $\phantom{\mathtt{INPUT} : }$  C is a circuit  given as a matrix product $C=\prod_{k=1}^{\ell}M_k$, of $\ell$ quantum gates
    $\phantom{\mathtt{INPUT} : }$  in the set $\{ P_i,Z_{\{i,j\}},X_{[ij]}\mid 0\leq i,j\leq n-1\}$,

    $\phantom{\mathtt{INPUT} : }$ $F_{\text{in}}$ is a circuit which is already in \pzx.

    $\mathtt{OUTPUT} : F_{\text{out}}$ is a circuit equivalent to the product $CF_{\mathrm{in}}$,
    
    $\phantom{\mathtt{OUTPUT} : }$ written in \pzx $Z_{\vct{v}} P_{\vct{b}}Z_{\mat{B}}X_{\mat{A}}$.

    $\mathtt{1}\quad\ $/* \emph{initialisation of the form} $F_{\text{out}}$ */ 

    $\mathtt{2}\quad\ F_{\text{out}}\leftarrow F_{\text{in}};$ 

    $\mathtt{3}\quad\ \mathtt{for}\ k=\ell\ \mathtt{to}\ 1\ \mathtt{do}$

    $\mathtt{4}\quad\quad\ $/* \emph{Case a : $M_k$ is a $\cz$ gate} */
    
    $\mathtt{5}\quad\quad\ \mathtt{if}\  M_k=Z_{\{i,j\}}\ \mathtt{then}$
    
    $\mathtt{6}\quad\quad\quad\  \mat B\leftarrow \mat B \oplus\{\{i,j\}\}\ ; $

    $\mathtt{7}\quad\quad\ $/* \emph{Case b : $M_k$ is a $\PP$ gate} */    

    $\mathtt{8}\quad\quad\ \mathtt{else\ if\ } M_k= P_i\ \mathtt{then}$
    
    $\mathtt{9}\quad\quad\quad\ \vct v\leftarrow \vct v\oplus b_i\vct{e}_i\ ; \ \vct{b}\leftarrow \vct{b} \oplus \vct{e}_i\ ; $

    $\mathtt{10}\quad\quad $/* \emph{Case c : $M_k$ is a $\cnot$ gate} */
    
    $\mathtt{11}\quad\quad\mathtt{else\ }\ $
    
    $\mathtt{12}\quad\quad\quad \vct v\leftarrow [ji]\vct v \oplus b_ib_j\vct{e}_j\oplus b_{ij}\vct{e}_j\ ;$

    $\mathtt{13}\quad\quad\quad \mat B\leftarrow [ji]\mat B[ij]\oplus b_i\{\{i,j\}\}\ ;$

    $\mathtt{14}\quad\quad\quad \vct{b}\leftarrow [ji]\vct{b}\ ;\ \mat A\leftarrow [ij]\mat A\ ;$
    
    $\mathtt{15}\quad\mathtt{return}\ F_{\text{out}};$
    
    \caption{ Algorithm  $\ctopzx$ : the time complexity of this algorithm is only  $O(n\ell)$ since we use row and column additions instead of matrix multiplication in Case c (thanks to Proposition \ref{tij-mult}). At the end of the algorithm, the matrix $\mat A$ is the product of all the transvections corresponding to the $\cnot$ gates that appear in the input circuit $C$, in the same order.\label{C-to-PZX}}
  \end{figure}

\begin{theo}[\textbf {The \pzx for a quantum circuit of $\PP,\cz$ and $\cnot$ gates}\label{decompo}]
  Any element of $\czxpg$ admits a unique decomposition in the form
  \begin{equation}
    Z_{\vct{v}} P_{\vct{b}}Z_{\mat{B}}X_{\mat{A}},\label{nf1}
  \end{equation}
  where $\vct v,\vct b\in\F^n$, $\mat B\in\BG$, $\mat A\in\GL$.
  
  The group $\czxpg$ is the semidirect product of its normal subgroup $\czpg$ with $\cnotg$, \textit{i.e.} $\czxpg=\czpg\rtimes\cnotg$.
  The order of $\czxpg$ is therefore $2^{n(n+1)}\prod_{i=1}^n(2^i-1)$. 
\end{theo}

\begin{proof}
  The existence of the decomposition can be proved by using the $\ctopzx$ algorithm described in Figure \ref{C-to-PZX} :
  let $\ell\geq 0$ be an integer and $C= \prod_{k=1}^{\ell}M_k$ be an element of $\czxpg$, where $M_{k}$ is a unitary in the gate set $\{ P_i,Z_{ij},X_{[ij]}\mid 0\leq i,j\leq n-1\}$. Then the form $ Z_{\vct{v}} P_{\vct{b}}Z_{\mat{B}}X_{\mat{A}}$ for $C$ is the result of
  Algorithm $\ctopzx$  applied to $C$ and $F_{\mathrm{in}}=I$,
  that is : $Z_{\vct{v}} P_{\vct{b}}Z_{\mat{B}}X_{\mat{A}}=\ctopzx(C, I)$.  
  Let us prove, by contradiction, the unicity of this decomposition.
  Suppose that $ Z_{\vct{v}} P_{\vct{b}}Z_{\mat{B}}X_{\mat{A}}= Z_{\vct{v}'} P_{\vct{b}'}Z_{\mat{B}'}X_{\mat{A}'}$.
  If $\mat A\neq \mat{A}'$, there exists $\vct x\in\F^n$ such that $\mat A\vct x\neq \mat A'\vct x$. But this leads to a contradiction
  because $ Z_{\vct{v}} P_{\vct{b}}Z_{\mat{B}}X_{\mat{A}}\ket{\vct x}= Z_{\vct v'} P_{\vct b'}Z_{\mat B'}X_{\mat A'}\ket{\vct x}$, so $ Z_{\vct{v}} P_{\vct{b}}Z_{\mat{B}}\ket{\mat A\vct x}= Z_{\vct v'} P_{\vct b'}Z_{\mat B'}\ket{\mat A'\vct x}$, hence $\ket{\mat A\vct x}$ and $\ket{\mat A'\vct x}$ are two different basis vector that are collinear, which is impossible. So $\mat A=\mat A'$ and $ Z_{\vct{v}} P_{\vct{b}}Z_{\mat{B}}= Z_{\vct v'} P_{\vct b'}Z_{\mat B'}$.
  If $\vct b\neq \vct b'$, we can suppose, without loss of generality, that there exists $i$ such that $b_i=1$ and $b'_i=0$. Then $ P_{\vct{b}}\ket{\vct e_i}=\ii\ket{\vct e_i}$   and $ P_{\vct b'}\ket{\vct e_i}=\ket{\vct e_i}$, so $\ii Z_{\vct{v}}Z_{\mat{B}}\ket{\vct e_i}= Z_{\vct v'}Z_{\mat B'}\ket{\vct e_i}$, hence $\ii\ket{\vct e_i}=\pm\ket{\vct e_i}$, which is not possible. Thus $\vct b=\vct b'$. Finally, if $ Z_{\vct{v}}Z_{\mat{B}}= Z_{\vct v'}Z_{\mat B'}$, we show that $\vct v=\vct v'$ and $\mat B=\mat B'$ by comparing their action on $\ket{\vct e_i}$ and $\ket{\vct e_i\oplus \vct e_j}$ for any $i,j$.

  Since $\czpg$ is a normal subgroup of $\czxpg$, the semidirect product structure is a consequence of the existence and uniqueness of the decomposition. The order of $\czxpg$ is computed using Theorem \ref{X-GL}.  
\end{proof}

\subsection{Toolbox of conjugation rules} 
The  $\ctopzx$ algorithm computes a \nf for particular stabilizer circuits, consisting only of $\PP, \cz$ and $\cnot$ gates.
  In the next section, we use the $\ctopzx$ algorithm as a subroutine called by the main algorithm that computes a \nf for any stabilizer circuit.
  In order to describe this algorithm,  we need some more conjugation rules. 
Let $\h=\prod_{i=0}^{n-1} H_i$, let $U^{\h}= \h U \h^{-1}= \h U \h$ for any $U\in\UG$. We use $U^{- \h}$ as a shorthand notation for $(U^{ \h})^{-1}$. The following identities hold.
\begin{align}
  & X_{\mat{A}}^{\h} =X_{\mat A^{-\T}}\label{conj-XA-h}\\
  & P_i^{ \h} P_i P_i^{- \h}=\ee^{\ii\frac{\pi}{4}} H_i X_i\label{conj-p-ph}\\
  & P_i^{ \h}Z_{\{i,k\}} P_i^{- \h}=Z_{\{i,k\}}X_{[ik]} P_k\label{conj-z-ph}\\
  &Z_{\{i,j\}}^{ \h} P_jZ_{\{i,j\}}^{ -\h}= P_i^{ \h}X_{[ij]} P_j\label{conj-pj-zijh}\\
  &Z_{\{i,j\}}^{ \h}Z_{\{i,j\}}Z_{\{i,j\}}^{ -\h}=Z_{\{i,j\}}Z_{\{i,j\}}^{ \h}Z_{\{i,j\}}= H_i H_jX_{(ij)}=X_{(ij)} H_i H_j\label{conj-zij-zijh}\\
  &Z_{\{i,j\}}^{ \h}Z_{\{i,k\}}Z_{\{i,j\}}^{ -\h}=X_{[jk]}Z_{\{i,k\}}=Z_{\{i,k\}}X_{[jk]}\quad (i,j,k\text{ distinct})\label{conj-zik-zijh}
 \end{align}

To write a stabilizer circuit in normal form, we also need the conjugation rules of a Pauli product of type $X_{\vct u} Z_{\vct v}$
by the gates $P_i$, $X_{[ij]}$, $Z_{\{i,j\}}$ and $\h$.
\begin{align}
  & P_i X_{\vct u} Z_{\vct v} P_i^{-1}=\ii^{u_i} X_{\vct u} Z_{\vct v\oplus u_i \vct e_i}\label{conj-xz-pi}\\
  &X_{[ij]} X_{\vct u} Z_{\vct v}X_{[ij]}= X_{[ij]\vct u} Z_{[ji]\vct v}\label{conj-xz-xij}\\
  &Z_{\{i,j\}} X_{\vct u} Z_{\vct v}Z_{\{i,j\}}=(-1)^{u_i u_j} X_{\vct u} Z_{\vct v\oplus u_j \vct e_i\oplus u_i \vct e_j}=(-1)^{u_i u_j} X_{\vct u} Z_{\vct v\oplus \{\{i,j\}\}\vct u}\label{conj-xz-zij}\\
  & \h X_{\vct{u}} Z_{\vct{v}} \h= X_{\vct{v}} Z_{\vct{u}}\label{conj-xz-h}
\end{align}

The proofs of Identities \eqref{conj-XA-h} to \eqref{conj-xz-h} are short and simple. They are based on the different conjugation rules already seen in
this section and on the observation that, to conjugate by $\h$, it is just necessary to take in account the indices concerned by the operation (for instance $Z_{\{i,j\}}^{ \h}=H_iH_jZ_{\{i,j\}}H_jH_i$).
We indicate in the tables below the formulas to be used to prove each identity. 
\begin{center}
\begin{tabular}{|c|c|}\hline
  Identity\dots& comes from\dots\\\hline
  \ref{conj-XA-h}& \ref{conj-xij-h}\\\hline
  \ref{conj-p-ph}&\ref{hp3}, \ref{conj-z-h}, \ref{conj-x-p}, \ref{yixz}\\\hline
  \ref{conj-z-ph}&\ref{conj-Pi}, \ref{hp3}, \ref{zijxij}\\\hline
  \ref{conj-pj-zijh}&\ref{zijxij}, \ref{conj-Pi}\\\hline
  \ref{conj-zij-zijh}& \ref{zijxij}, \ref{sijxij}\\\hline
  \end{tabular}\qquad
\begin{tabular}{|c|c|}\hline
  Identity\dots& comes from\dots\\\hline
  \ref{conj-zik-zijh} & \ref{zijxij}, \ref{conj-xij-xjk}\\\hline
  \ref{conj-xz-pi}&\ref{conj-x-p}, \ref{yixz}\\\hline
  \ref{conj-xz-xij}&\ref{conj-Za-xij}, \ref{conj-z-h}, \ref{conj-xij-h}\\\hline
  \ref{conj-xz-zij}&\ref{conj-xz-xij}, \ref{anticom}, \ref{conj-z-h}\\\hline
  \ref{conj-xz-h}&\ref{conj-z-h}\\\hline
\end{tabular}
\end{center}

Finally, we obtain more general conjugation rules by unitaries of type  $P_{\vct{b}}, X_{\mat A}$ or $Z_{\mat B}$ by iterating Identities \eqref{conj-xz-pi}, \eqref{conj-xz-xij} and \eqref{conj-xz-zij}.
\begin{align}
  & P_{\vct{b}} X_{\vct u} Z_{\vct v} P_{\vct{b}}^{-1}=\ii^{\sum u_i} X_{\vct u} Z_{\vct v\oplus\sum b_i u_i \vct e_i}=
    \ii^{\sum u_i} X_{\vct u} Z_{\vct v\oplus \vct b\odot \vct u}\label{conj-xz-pb}\\
  &X_{\mat A} X_{\vct u} Z_{\vct v}X_{\mat A}^{-1}= X_{\mat A\vct u} Z_{\mat A^{-\T}\vct v}\label{conj-xz-XA}\\
  &Z_{\mat B} X_{\vct u} Z_{\vct v}Z_{\mat B}=(-1)^{q_{\mat B}(\vct u)} X_{\vct u} Z_{\vct v\oplus \mat B\vct u}\label{conj-xz-ZB}
\end{align}

\section{The generalized \pzx for stabilizer circuits\label{NF}}
In this section we provide an algorithm to compute a \nf for stabilizer circuits. 
This form is a generalization of the \pzx  ($ Z_{\vct v} P_{\vct{b}}Z_{\mat{B}}X_{\mat{A}}$) introduced in the previous section \eqref{nf1}.
\begin{defi}
  A unitary matrix $C$ corresponding to a circuit of Clifford gates is in \gpzx  when it is written in the form 
\begin{equation}
  C=\ee^{\ii\phi}H_{\mathbf{r}} Z_{\mathbf{u}}P_{\mathbf{d}}Z_{\mathbf{D}} H_{\mathbf{s}} Z_{\mathbf{v}}P_{\mathbf b}Z_{\mathbf{B}}X_{\mathbf{A}},\label{nf-formula}
\end{equation}
where     $\vct r, \vct u,\vct d, \vct s, \vct v, \vct b$ are vectors in $\F^n$, $\mat D$ and $\mat B$ are matrices in $\BG$, $\mat A$ is an invertible matrix in $\GL$, and $\phi\in\{k\frac{\pi}{4},k\in\Z\}$.
\end{defi}
  We remark that, unlike the PZX form, the \gpzx is not unique. Indeed, from Identity \ref{conj-zij-zijh}, one has
  $ H_i H_jZ_{\{i,j\}} H_i H_jZ_{\{i,j\}}=Z_{\{i,j\}} H_i H_jX_{(ij)}$.
 
  \subsection{Stability properties of an intermediate form}

  We introduce a form for stabilizer circuits that we use as an intermediary technical step to compute the GenPZX form  and we prove three lemmas concerning this intermediate form.
  \begin{defi}
A unitary matrix $C$ corresponding to a circuit of Clifford gates is in \emph{\intnf\ } when it is written in the form 
\begin{equation}
C= H_{\vct a} P_{\vct{d}}Z_{\mat{D}} \h\ee^{\ii\phi} X_{\vct u} Z_{\vct v} P_{\vct{b}}Z_{\mat{B}}X_{\mat{A}},\label{pre-nf}
\end{equation}
where $\vct a, \vct d, \vct u, \vct v, \vct b$ are vectors in $\F^n$, $\mat D$ and $\mat B$ are matrices in $\BG$, $\mat A$ is an invertible matrix in $\GL$, and $\phi\in\{k\frac{\pi}{4},k\in\Z\}$ 
\end{defi}

The first lemma is quiet obvious and we write it just to keep our results consistent. The other two lemmas are more technical because we need to distinguish many cases and there are many variables. However, the calculations are simple, essentially based on the different conjugations rules of Section \ref{groups}. In order to make the reading easier, we use two colors : the \red{red} color to emphasize a part of an expression which is already in \intnf\ and the \blue{blue} color  to point out a part of an expression modified by the current computation or to indicate the next gates that we want to merge in the \intnf. We also use dots ($\cdot$) to separate blocks of unitary matrices.

\begin{lem}\label{stab-H}
  The \intnf\ \eqref{pre-nf} is stable by left multiplication by a Hadamard gate : if a unitary matrix $C$ is in  \intnf, then $H_iC$ can be written in  \intnf, for any $i=0\dots n-1$.
\end{lem}

\begin{proof}
\begin{align*}
H_iC&=\blue{H_{\vct e_i}}\red{H_{\vct a}P_{\vct{d}}Z_{\mat{D}}\ee^{\ii\phi} \h X_{\vct u} Z_{\vct v} P_{\vct{b}}Z_{\mat{B}}X_{\mat{A}}}\\
&=\red{H_{\vct a\oplus e_i} P_{\vct{d}}Z_{\mat{D}}\ee^{\ii\phi} \h X_{\vct u} Z_{\vct v} P_{\vct{b}}Z_{\mat{B}}X_{\mat{A}}}
\end{align*}
So $H_iC$ is in \intnf.
\end{proof}

\begin{lem}\label{stab-P}
  The \intnf\ \eqref{pre-nf} is stable by left multiplication by a phase gate : if a unitary matrix $C$ is in  \intnf, then $P_iC$ can be written in \intnf, for any $i=0\dots n-1$.
\end{lem}
\begin{proof}
\begin{align*}
P_iC&=\blue{P_{i}}\red{H_{\vct a}P_{\vct{d}}Z_{\mat{D}}\ee^{\ii\phi} \h X_{\vct u} Z_{\vct v} P_{\vct{b}}Z_{\mat{B}}X_{\mat{A}}}\\
&= H_{\vct a}\cdot\blue{H_{\vct a} P_i H_{\vct a}}\cdot\red{P_{\vct{d}}Z_{\mat{D}} \h\ee^{i\phi} X_{\vct u} Z_{\vct v} P_{\vct{b}} Z_{\mat{B}}X_{\mat{A}}}
\end{align*}
We distinguish 2 cases, according to the possible values of $a_i$.\medskip

\texttt{Case 1} : $a_i=0$. In this case $H_{\vct a} P_i H_{\vct a}= P_i$, so
\begin{align*}
  P_iC&=H_{\vct a}\cdot\blue{P_i}\cdot\red{P_{\vct{d}}Z_{\mat{D}} \h\ee^{i\phi} X_{\vct u} Z_{\vct v} P_{\vct{b}} Z_{\mat{B}}X_{\mat{A}}}\\
                   &\stackrel{\eqref{czpg-mult}}{=} H_{\vct a} \blue{Z_{d_i \vct e_i}} \red{P_{\vct d\oplus \vct e_i}Z_{\mat{D}} \h\ee^{\ii\phi} X_{\vct u} Z_{\vct v} P_{\vct{b}}Z_{\mat{B}}X_{\mat{A}}}\\
                   &\stackrel{\eqref{conj-xz-h}}{=} \red{H_{\vct a} P_{\vct d\oplus \vct e_i}Z_{\mat{D}} \h\ee^{\ii\phi} X_{\vct u\oplus d_i \vct{e}_i} Z_{\vct v} P_{\vct{b}}Z_{\mat{B}}X_{\mat{A}}}
\end{align*}
and $ P_iC$ is in \intnf.\medskip

\texttt{Case 2} : $a_i=1$. In this case, $H_{\vct a} P_i H_{\vct a}= P_i^{\h}$, so
\begin{equation}
  P_iC= H_{\vct a}\cdot\blue{ P_i^{\h}}\cdot\red{P_{\vct{d}}Z_{\mat{D}} \h\ee^{i\phi} X_{\vct u} Z_{\vct v} P_{\vct{b}} Z_{\mat{B}}X_{\mat{A}}}\label{case-22}.
\end{equation}
We use many times conjugation by ${P_i^{ \h}}$ or $P_i$  in order to merge $P_i$ with $P_{\vct b}$ :
\begin{align*}
  P_iC&=H_{\vct a}\cdot \blue{P_i^{ \h}} P_{\vct d} \blue{P_i^{- \h}}\cdot \blue{P_i ^{\h}}Z_{\mat{D}} \blue{P_i^{-\h}}\cdot \h\cdot  \blue{P_i}\ee^{\ii\phi} X_{\vct u} Z_{\vct v} \blue{P_i^{-1}}\cdot  \blue{P_i} \red{P_{\vct{b}}Z_{\mat{B}}X_{\mat{A}}}\\
      &\stackrel{\eqref{czpg-mult}}{=} H_{\vct a}\cdot P_i^{\h} P_{\vct d} P_i^{-\h}\cdot P_i^{\h} Z_{\mat{D}} P_i^{-\h}\cdot
        \h\cdot P_i\ee^{\ii\phi} X_{\vct u} Z_{\vct v} P_i^{-1}  \cdot \blue{Z_{b_i \vct e_i}}\red{P_{\vct b\oplus \vct e_i}Z_{\mat{B}}X_{\mat{A}}}\\
      &\stackrel{\eqref{conj-xz-pi}}{=} H_{\vct a}\cdot P_i^{\h} P_{\vct d} P_i^{-\h}\cdot P_i^{\h} Z_{\mat{D}} P_i^{-\h}\cdot
        \h\cdot\blue{\ee^{\ii\phi}\ii^{u_i} X_{\vct u} Z_{\vct v\oplus u_i \vct e_i}\cdot Z_{b_i \vct e_i}}\red{P_{\vct b\oplus \vct e_i}Z_{\mat{B}}X_{\mat{A}}}
\end{align*}
Let $\phi'=\phi + u_i\frac\pi2$, $\vct u'=\vct u$, $\vct v'=\vct v\oplus u_i \vct e_i\oplus b_i \vct e_i$ and $\vct b'=\vct b\oplus \vct e_i$ then
\begin{equation*}
P_iC= H_{\vct a}\cdot P_i^{\h} P_{\vct d} P_i^{-\h}\cdot P_i^{\h} Z_{\mat{D}} P_i^{-\h}\cdot
\red{\h\ee^{\ii\phi'} X_{\vct u'} Z_{\vct v'}P_{\vct b'}Z_{\mat{B}}X_{\mat{A}}}.
\end{equation*}
We need to distinguish two subcases, according to the values of $d_i$.
\medskip

\texttt{Case 2.1} : $d_i=0$.  In this case, $P_i^{\h} P_{\vct d} P_i^{-\h}= P_{\vct d}$, so
  \begin{equation}
    P_iC= H_{\vct a}\blue{P_{\vct d}}\cdot P_i^{\h} Z_{\mat{D}} P_i^{-\h}\cdot
    \red{\h\ee^{\ii\phi'} X_{\vct u'} Z_{\vct v'}P_{\vct b'}Z_{\mat{B}}X_{\mat{A}}}.\label{case-221}
\end{equation}
Let $\mat D_i=\{\{p,q\}\in \mat D\mid i\in\{p,q\}\}$,
then $P_i^{\h} Z_{\mat{D}} P_i^{-\h}=Z_{\mat D\oplus \mat D_i}\cdot P_i^{\h}Z_{\mat{D}_i}P_i^{-\h}$.

Let $\Lambda_i=\{k\mid \{k,i\}\in D\}$, then

$P_i^{\h} Z_{\mat{D}} P_i^{-\h}=Z_{\mat D\oplus \mat D_i} \cdot \prod_{k\in \Lambda_{i}}P_i^{\h}Z_{\{i,k\}}P_i^{-\h}
\stackrel{\eqref{conj-z-ph}}{=}Z_{\mat D\oplus \mat D_i}\cdot \prod_{k\in \Lambda_{i}}Z_{\{i,k\}}X_{[ik]} P_k$, hence
\begin{equation*}
P_iC= H_{\vct a} P_{\vct d}\cdot \blue{Z_{\mat D\oplus \mat D_i}\prod_{k\in \Lambda_{i}}Z_{\{i,k\}}X_{[ik]} P_k}\cdot
\red{ \h\ee^{\ii\phi'} X_{\vct u'} Z_{\vct v'} P_{\vct b'}Z_{\mat{B}}X_{\mat{A}}}.
\end{equation*}
We apply the $\ctopzx$ algorithm with parameters $C=\prod_{k\in \Lambda_{i}}Z_{\{i,k\}}X_{[ik]} P_k$ and $F_{\text{in}}=I$ :
let $Z_{\vct w'} P_{\vct d'}Z_{\mat D'}X_{\mat A'}=\ctopzx(\prod_{k\in \Lambda_{i}}Z_{\{i,k\}}X_{[ik]} P_k, I)$, then $A'=\prod_{k\in \Lambda_{i}}[ik]$ (see Figure \ref{C-to-PZX}) and
\begin{align*}
  P_iC&= H_{\vct a}P_{\vct d}Z_{\mat D\oplus \mat D_i} \blue{Z_{\vct w'} P_{\vct d'}Z_{\mat D'}X_{\mat A'}}\red{ \h\ee^{\ii\phi'} X_{\vct u'} Z_{\vct v'} P_{\vct b'}Z_{\mat{B}}X_{\mat{A}}}\\
      &\stackrel{\eqref{czpg-mult}}{=} H_{\vct a}\cdot\blue{Z_{\vct w' \oplus \; \vct d \odot \vct d'}P_{\vct d\oplus \vct d'}Z_{\mat D\oplus \mat D_i\oplus \mat D'} }\cdot \blue{X_{\mat A'}}\red{ \h\ee^{\ii\phi'} X_{\vct u'} Z_{\vct v'} P_{\vct b'}Z_{\mat{B}}X_{\mat{A}}}\\
      &\stackrel{\eqref{conj-XA-h}}{=} H_{\vct a} Z_{\vct w'\oplus \; \vct d\odot\vct d'}P_{\vct d\oplus \vct d'}Z_{\mat D\oplus \mat D_i\oplus \mat D'}   \h\cdot\ee^{\ii\phi'}\blue{X_{\mat A'^{-\T}}} X_{\vct u'} Z_{\vct v'}\blue{X_{\mat A'^{-\T}}^{-1}}\cdot \blue{X_{\mat A'^{-\T}}} \red{P_{\vct b'}Z_{\mat{B}}X_{\mat{A}}},
\end{align*}
where $\mat A'^{-\T}=\prod_{k\in \Lambda_{i}}[ki]$. Using the decomposition of $\mat A'^{-\T}$ in transvections, we iterate Identity \eqref{conj-xz-xij} on the Pauli block 
$X_{\vct u'} Z_{\vct v'}$ and computes thereby two vectors $\vct u''$ and $\vct v''$ such that $ X_{\vct u''} Z_{\vct v''}=X_{\mat A'^{-\T}}X_{\vct u'} Z_{\vct v'}X_{\mat A'^{-\T}}^{-1}$. Then, we apply the $\ctopzx$ algorithm with parameters $C=X_{\mat A'^{-\T}}=\prod_{k\in \Lambda_{i}}X_{[ki]}$ and $F_{\text{in}}=P_{\vct b'}Z_{\mat{B}}X_{\mat{A}} : $
let $Z_{\vct w''} P_{\vct b''}Z_{\mat B''}X_{\mat A''}=\ctopzx\left(X_{\mat A'^{-\T}}, P_{\vct b'}Z_{\mat{B}}X_{\mat{A}}\right)$, we obtain
\begin{align*}
P_iC&= H_{\vct a}Z_{\vct w'\oplus \; \vct d\odot\vct d'} P_{\vct d\oplus \vct d'}Z_{\mat D\oplus \mat D_i\oplus \mat D'}\h\cdot\ee^{\ii\phi'}\blue{ X_{\vct u''} Z_{\vct v''}} \cdot \red{Z_{\vct w''} P_{\vct b''}Z_{\mat B''}X_{\mat A''}}\\
&\stackrel{\eqref{conj-xz-h}}{=} \red{H_{\vct a} P_{\vct d\oplus \vct d'}Z_{\mat D\oplus \mat D_i\oplus \mat D'}\h\ee^{\ii\phi'} X_{\vct u''\oplus \vct w'\oplus \; \vct d\odot\vct d'} Z_{\vct v''\oplus \vct w''} P_{\vct b''}Z_{\mat B''}X_{\mat A''}}, 
\end{align*}
and $ P_iC$ is in  \intnf.
\medskip

\texttt{Case 2.2} :  $d_i=1$.  In this case, $P_i^{\h} P_{\vct d} P_i^{-\h}= P_i^{\h} P_i P_i^{-\h}\cdot P_{\vct d\oplus \vct e_i}
\stackrel{\eqref{conj-p-ph}}{=}\ee^{\ii\frac{\pi}{4}} H_i X_i\cdot P_{\vct d\oplus \vct e_i}$, so
\begin{align*}
  P_iC&= H_{\vct a}\cdot\blue{\ee^{\ii\frac{\pi}{4}} H_i X_i P_{\vct d\oplus \vct e_i}}\cdot P_i^{\h} Z_{\mat{D}} P_i^{-\h}\cdot
\red{\h\ee^{\ii\phi'} X_{\vct u'} Z_{\vct v'}P_{\vct b'}Z_{\mat{B}}X_{\mat{A}}}\\
      &= \blue{H_{\vct a\oplus \vct e_i} P_{\vct d\oplus \vct e_i} X_i}\cdot P_i^{\h} Z_{\mat{D}} P_i^{-\h}\cdot \red{\h\ee^{\ii(\phi'+\frac{\pi}{4})} X_{\vct u'} Z_{\vct v'}P_{\vct b'}Z_{\mat{B}}X_{\mat{A}}}\\
      &= H_{\vct a\oplus \vct e_i} P_{\vct d\oplus \vct e_i}\cdot P_i^{\h} Z_{\mat{D}} P_i^{-\h}\cdot\blue{C'}\cdot\red{\h\ee^{\ii(\phi'+\frac{\pi}{4})} X_{\vct u'} Z_{\vct v'}P_{\vct b'}Z_{\mat{B}}X_{\mat{A}}},
\end{align*}
where $C'=( P_i^{\h} Z_{\mat{D}} P_i^{-\h})^{-1} X_i( P_i^{\h} Z_{\mat{D}} P_i^{-\h})=P_i^{\h}Z_{\mat{D}}P_i^{-\h}X_iP_i^{\h} Z_{\mat{D}} P_i^{-\h}$.

Let us reduce $C'$ : 
as $P_i^{-\h}=\h P_i^{-1}\h=\h P_i Z_i \h\stackrel{\eqref{conj-xz-h}}{=}P_i^{\h}X_i=X_iP_i^{\h}$, we obtain 

$C'=P_i^{\h} Z_{\mat{D}}P_i^{\h}P_i^{\h} Z_{\mat{D}} X_iP_i^{\h}
=P_i^{\h} Z_{\mat{D}}Z_i^{\h}Z_{\mat{D}} X_iP_i^{\h}
\stackrel{\eqref{conj-xz-h}}{=}P_i^{\h} Z_{\mat{D}}X_iZ_{\mat{D}} X_iP_i^{\h}$.

Applying Identity \eqref{conj-xz-ZB} where $\vct u=\vct e_i$ and $\vct v=\mathbf{0}$, we get
$Z_{\mat{D}}X_iZ_{\mat{D}} =X_iZ_{\mat D \vct e_i}$, hence $C'=P_i^{\h} X_iZ_{\mat D \vct e_i}X_iP_i^{\h}$. As $\mat D \vct e_i$ is the column $i$ of matrix $\mat D$, the $i$-th bit of the vector $\mat D \vct e_i$  is 0, so $Z_{\mat D \vct e_i}$ commutes with $X_i$ and 
$X_{\mat D \vct e_i}$ commutes with $P_i$. Hence $C'=P_i^{\h} Z_{\mat D \vct e_i}P_i^{\h}=\h P_i \h Z_{\mat D \vct e_i}\h P_i \h
\stackrel{\eqref{conj-xz-h}}{=}\h P_i X_{\mat D \vct e_i} P_i \h
=\h X_{\mat D \vct e_i}Z_i\h$. So
\begin{align*}
  P_iC&=H_{\vct a\oplus \vct e_i} P_{\vct d\oplus \vct e_i}\cdot P_i^{\h} Z_{\mat{D}} P_i^{-\h}\cdot \blue{\h X_{\mat D \vct e_i}Z_i\h}\cdot
      \red{\h\ee^{\ii(\phi'+\frac{\pi}{4})} X_{\vct u'} Z_{\vct v'}P_{\vct b'}Z_{\mat{B}}X_{\mat{A}}}\\
      &=H_{\vct a\oplus \vct e_i} P_{\vct d\oplus \vct e_i}\cdot P_i^{\h} Z_{\mat{D}} P_i^{-\h}\cdot \blue{\h X_{\mat D \vct e_i}Z_i}\red{\ee^{\ii(\phi'+\frac{\pi}{4})} X_{\vct u'} Z_{\vct v'}P_{\vct b'}Z_{\mat{B}}X_{\mat{A}}}
\end{align*}
We merge $X_{\mat D \vct e_i}Z_i$ into the Pauli block $\ee^{\ii(\phi'+\frac{\pi}{4})}X_{\vct u'} Z_{\vct v'}$ by using Identity \eqref{pauli-mult} and obtain thereby a phase $\phi''$ and two vectors $\vct u''$ and $\vct v''$ such that 
\begin{equation}
 P_iC=H_{\vct a\oplus \vct e_i} P_{\vct d\oplus \vct e_i}\cdot P_i^{\h} Z_{\mat{D}} P_i^{-\h} 
\cdot\red{\h\ee^{\ii\phi''} X_{\vct u''} Z_{\vct v''}P_{\vct b'}Z_{\mat{B}}X_{\mat{A}}}.\label{case-222}
\end{equation}

We observe that Equality \ref{case-222} has the same form as Equality \ref{case-221}. Therefore, to write $P_iC$ in  \intnf, one can proceed as in Case 2.1, starting from Equality \ref{case-221}.
\end{proof}

\begin{lem}\label{stab-X}
  The \intnf \eqref{pre-nf} is stable by left multiplication by a $\cnot$ gate : if a unitary matrix $C$ is in  \intnf, then $X_{[ij]}C$ can be written in \intnf, for any $i,j=0\dots n-1$, $i\neq j$.
\end{lem}
\begin{proof}
\begin{align*}
X_{[ij]}C&= \blue{X_{[ij]}}\red{H_{\vct a}P_{\vct{d}}Z_{\mat{D}} \h\ee^{i\phi} X_{\vct u} Z_{\vct v} P_{\vct{b}}Z_{\mat{B}}X_{\mat{A}}}\\
&= H_{\vct a}\cdot\blue{ H_{\vct a}X_{[ij]} H_{\vct a}}\cdot\red{ P_{\vct{d}}Z_{\mat{D}} \h\ee^{i\phi} X_{\vct u} Z_{\vct v} P_{\vct{b}}Z_{\mat{B}}X_{\mat{A}}}.
\end{align*}
We need to distinguish 4 cases, according to the values of $(a_i,a_j)$.\medskip

\texttt{Case 1} : $(a_i,a_j)= (0,0)$. In this case, $ H_{\vct a}X_{[ij]} H_{\vct a}=X_{[ij]}$, so
\begin{equation}
X_{[ij]}C= H_{\vct a}\cdot \blue{X_{[ij]}} \cdot \red{P_{\vct{d}}Z_{\mat{D}}  \h\ee^{i\phi} X_{\vct u} Z_{\vct v} P_{\vct{b}}Z_{\mat{B}}X_{\mat{A}}}.\label{case-31}
\end{equation}

We apply the $\ctopzx$ algorithm with parameters $C=X_{[ij]}$ and $F_{\text{in}}=P_{\vct{d}}Z_{\mat{D}}$ :
let $ Z_{\vct v'} P_{\vct d'}Z_{\mat D'}X_{[ij]}=\ctopzx(X_{[ij]}, P_{\vct{d}}Z_{\mat{D}})$ , then
\begin{align*}
  X_{[ij]}C&= H_{\vct a}\cdot\blue{Z_{\vct v'} P_{\vct d'}Z_{\mat D'}X_{[ij]}}\cdot\red{ \h\ee^{i\phi} X_{\vct u} Z_{\vct v} P_{\vct{b}}Z_{\mat{B}}X_{\mat{A}}}\\
           &\stackrel{\eqref{conj-xij-h}}{=} H_{\vct a} Z_{\vct v'} P_{\vct d'}Z_{\mat D'}\h\cdot\blue{ X_{[ji]}}\ee^{i\phi} X_{\vct u} Z_{\vct v}\blue{X_{[ji]}}
             \cdot \blue{X_{[ji]}} \red{P_{\vct{b}}Z_{\mat{B}}X_{\mat{A}}}\\
           &\stackrel{\eqref{conj-xz-xij}}{=} H_{\vct a} Z_{\vct v'} P_{\vct d'}Z_{\mat D'} \h\cdot\blue{\ee^{i\phi} X_{[ji]\vct u} Z_{[ij]\vct v}\cdot X_{[ji]}}\red{ P_{\vct{b}}Z_{\mat{B}}X_{\mat{A}}}\\
           &\stackrel{\eqref{conj-xz-h}}{=} H_{\vct a} P_{\vct d'}Z_{\mat D'}\h\cdot\blue{\ee^{i\phi} X_{[ji]\vct u\oplus \vct v'} Z_{[ij]\vct v}\cdot X_{[ji]}}\red{ P_{\vct{b}}Z_{\mat{B}}X_{\mat{A}}}
 \end{align*}

 We apply the $\ctopzx$ algorithm with parameters $C=X_{[ji]}$ and $F_{\text{in}}=P_{\vct{b}}Z_{\mat{B}}X_{\mat{A}}$ :
 let $ Z_{\vct v''} P_{\vct b''}Z_{\mat B''}X_{\mat A''}=\ctopzx(X_{[ji]}, P_{\vct{b}}Z_{\mat{B}}X_{\mat{A}})$, then
 \begin{equation*}
   X_{[ij]}C= \red{H_{\vct a} P_{\vct d'}Z_{\mat D'} \h\ee^{i\phi} X_{[ji]\vct u\oplus \vct v'} Z_{[ij]\vct v\oplus \vct v''} P_{\vct b''}Z_{\mat B''}X_{\mat A''}},
\end{equation*}
and $X_{[ij]}C$ is in  \intnf.

\texttt{Case 2} : $(a_i,a_j)= (1,1)$. In this case, $ H_{\vct a}X_{[ij]} H_{\vct a}\stackrel{\eqref{conj-xij-h}}{=}X_{[ji]}$, so
we proceed as in Case 1, swapping $i$ and $j$.
\medskip

\texttt{Case 3} : $(a_i,a_j)= (1,0)$. In this case, $ H_{\vct a}X_{[ij]} H_{\vct a}\stackrel{\eqref{zijxij}}{=}Z_{\{i,j\}}$, so
\begin{align*}
  X_{[ij]}C&=H_{\vct a}\cdot\blue{ Z_{\{i,j\}}}\cdot\red{P_{\vct{d}}Z_{\mat D}  \h\ee^{i\phi} X_{\vct u} Z_{\vct v} P_{\vct{b}}Z_{\mat{B}}X_{\mat A}}\\
           &=\red{H_{\vct a} P_{\vct{d}}Z_{\mat D\oplus\{\{i,j\}\}}  \h\ee^{i\phi} X_{\vct u} Z_{\vct v} P_{\vct{b}}Z_{\mat{B}}X_{\mat A}},
\end{align*}
 and $X_{[ij]}C$ is in  \intnf.
\medskip

\texttt{Case 4} : $(a_i,a_j)= (0,1)$.

In this case, $ H_{\vct a}X_{[ij]} H_{\vct a}=H_jX_{[ij]}H_j\stackrel{\eqref{zijxij}}{=}H_jH_iZ_{\{i,j\}}H_iH_j=Z_{\{i,j\}}^{\h}$ and
\begin{equation*}
  X_{[ij]}C= H_{\vct a}\cdot\blue{Z_{\{i,j\}}^{\h}}\cdot\red{ P_{\vct{d}}Z_{\mat{D}} \h\ee^{i\phi} X_{\vct u} Z_{\vct v} P_{\vct{b}}Z_{\mat{B}}X_{\mat A}}.
  \end{equation*}

  Here the situation is more complicated because we need two distinguish different subcases, according to the possible values of
  $d_{ij}$ (the entry $(i,j)$ of matrix $\mat D$) and $(d_i,d_j)$ (the entries $i$ and $j$  of vector $\vct d$).\medskip
  
\texttt{Case 4.1}\label{case-41} : $d_{ij}=0$.
\begin{align*}
  X_{[ij]}C=& H_{\vct a}\cdot \blue{Z_{\{i,j\}}^{\h}} P_{\vct{d}}\blue{Z_{\{i,j\}}^{-\h}}\cdot \blue{Z_{\{i,j\}}^{\h}}Z_{\mat{D}}\blue{Z_{\{i,j\}}^{-\h}}\cdot\h \cdot \blue{Z_{\{i,j\}}} \ee^{i\phi}  X_{\vct u} Z_{\vct v}\blue{Z_{\{i,j\}}}\cdot \blue{Z_{\{i,j\}}} \red{P_{\vct{b}}Z_{\mat B}X_{\mat A}}\\
           \stackrel{\eqref{conj-xz-zij}}{=}& H_{\vct a}\cdot Z_{\{i,j\}}^{\h} P_{\vct{d}}Z_{\{i,j\}}^{\h}\cdot Z_{\{i,j\}}^{\h}Z_{\mat{D}}Z_{\{i,j\}}^{\h}\h \blue{ \ee^{i\phi}(-1)^{u_iu_j} X_{\vct u} Z_{\vct v\oplus\{\{i,j\}\}\vct u}}\red{ P_{\vct{b}}Z_{B\oplus \{\{i,j\}\}}X_{\mat A}}
  \end{align*}

Let $\phi'=\phi + u_iu_j\pi$, $\vct u'=\vct u$, $\vct v'=\vct v\oplus\{\{i,j\}\}\vct u$ and $\mat B'=\mat B\oplus \{\{i,j\}\}$, then
\begin{equation*}
X_{[ij]}C= H_{\vct a}\cdot Z_{\{i,j\}}^{\h} P_{\vct{d}}Z_{\{i,j\}}^{\h}\cdot Z_{\{i,j\}}^{\h}Z_{\mat{D}}Z_{\{i,j\}}^{\h}\cdot
 \red{\h \ee^{i\phi'} X_{\vct u'} Z_{\vct v'}P_{\vct{b}}Z_{\mat B'}X_{\mat A}}.
\end{equation*}
Let $\mat D_i=\{\{p,q\}\in \mat D\mid i\in\{p,q\}\}$  then $\mat D_i\cap \mat D_j=\emptyset$ (since $d_{ij}=0$) and 
$Z_{\{i,j\}}^{\h}Z_{\mat{D}}Z_{\{i,j\}}^{\h}=Z_{\mat D\oplus \mat D_i\oplus \mat D_j}\cdot Z_{\{i,j\}}^{\h}Z_{\mat D_i}Z_{\mat D_j}Z_{\{i,j\}}^{\h}$.\smallskip

Let $\Lambda_{i}=\{k\mid \{i,k\}\in \mat D\}$, then

$Z_{\{i,j\}}^{\h}Z_{\mat{D}}Z_{\{i,j\}}^{\h}=Z_{\mat D\oplus \mat D_i\oplus \mat D_j}\cdot Z_{\{i,j\}}^{\h}\cdot\prod_{k\in \Lambda_{i}}Z_{\{i,k\}}\cdot\prod_{k\in \Lambda_j}Z_{\{j,k\}}\cdot Z_{\{i,j\}}^{\h}$, so

$Z_{\{i,j\}}^{\h}Z_{\mat{D}}Z_{\{i,j\}}^{\h}\stackrel{\eqref{conj-zik-zijh}}{=}Z_{\mat D\oplus \mat D_i\oplus \mat D_j}\cdot\prod_{k\in \Lambda_{i}}Z_{\{i,k\}}X_{[jk]}\cdot\prod_{k\in \Lambda_j}Z_{\{j,k\}}X_{[ik]}$.\smallskip

Note that no phase gate is created when one applies the $\ctopzx$ algorithm to a sequence of $\cnot$ and $\cz$ gates as in the expression of $Z_{\{i,j\}}^{\h}Z_{\mat{D}}Z_{\{i,j\}}^{\h}$ above (see Figure \ref{C-to-PZX}) : let $ Z_{\vct w'}Z_{\vct D'}X_{\mat A'}=\ctopzx(Z_{\{i,j\}}^{\h}Z_{\mat{D}}Z_{\{i,j\}}^{\h}, I)$, then $\mat A'=\prod_{k\in \Lambda_{i}}[jk]\prod_{k\in \Lambda_j}[ik]$ and
\begin{align*}
X_{[ij]}C& = H_{\vct a}\cdot Z_{\{i,j\}}^{\h} P_{\vct{d}}Z_{\{i,j\}}^{\h}\cdot\blue{Z_{\vct w'}Z_{\mat D'}X_{\mat A'}} \cdot \red{\h\ee^{i\phi'} X_{\vct u'} Z_{\vct v'}P_{\vct{b}}Z_{\mat B'}X_{\mat A}}\\
         &\stackrel{\eqref{conj-XA-h}}{=} H_{\vct a}\cdot Z_{\{i,j\}}^{\h} P_{\vct{d}}Z_{\{i,j\}}^{\h}\cdot Z_{\vct w'}Z_{\mat D'}
           \cdot  \h\cdot \ee^{i\phi'}\blue{X_{\mat A'^{-\T}}}  X_{\vct u'} Z_{\vct v'}\blue{X_{\mat A'^{\T}}}\cdot \blue{X_{\mat A'^{-\T}}}\red{ P_{\vct{b}}Z_{\mat B'}X_{\mat A}},
\end{align*}
where $\mat A'^{-\T}=\prod_{k\in \Lambda_{i}}[kj]\prod_{k\in \Lambda_j}[ki]$.
Using the decomposition of $\mat A'^{-\T}$ in transvections, we iterate Identity \eqref{conj-xz-xij} on the Pauli block 
$X_{\vct u'} Z_{\vct v'}$ and computes thereby two vectors $\vct u''$ and $\vct v''$ such that $ X_{\vct u''} Z_{\vct v''}=X_{\mat A'^{-\T}}X_{\vct u'} Z_{\vct v'}X_{\mat A'^{-\T}}^{-1}$. Then we apply the $\ctopzx$ algorithm with parameters $C=X_{\mat A'^{-\T}}=\prod_{k\in \Lambda_{i}}X_{[kj]}\prod_{k\in \Lambda_j}X_{[ki]}$ and $F_{\text{in}}=P_{\vct{b}}Z_{\mat B'}X_{\mat A} : $
let $ Z_{\vct w''} P_{\vct b''}Z_{\mat B''}X_{\mat A''}=\ctopzx (X_{\mat A'^{-\T}}, P_{\vct{b}}Z_{\mat B'}X_{\mat A})$, we obtain
\begin{align*}
 X_{[ij]}C&= H_{\vct a}\cdot Z_{\{i,j\}}^{\h} P_{\vct{d}}Z_{\{i,j\}}^{\h}\cdot Z_{\vct w'}Z_{\mat D'}\cdot  \h\cdot
\ee^{i\phi'}\blue{ X_{\vct u''} Z_{\vct v''}}\cdot\red{ Z_{\vct w''} P_{\vct b''}Z_{\mat B''}X_{\mat A''}}\\
&\stackrel{\eqref{conj-xz-h}}{=} H_{\vct a}\cdot Z_{\{i,j\}}^{\h} P_{\vct{d}}Z_{\{i,j\}}^{\h}\cdot\red{ Z_{\mat D'}\h
\ee^{i\phi'} X_{\vct u''\oplus \vct w'} Z_{\vct v''\oplus \vct w''} P_{\vct b''}Z_{\mat B''}X_{\mat A''}}.
\end{align*}

\texttt{Case 4.1.1} : $(d_i,d_j)=(0,0)$. In this case, $Z_{\{i,j\}}^{\h} P_{\vct{d}}Z_{\{i,j\}}^{\h}= P_{\vct{d}}$, so
\begin{equation*}
X_{[ij]}C= \red{H_{\vct a} P_{\vct{d}} Z_{\mat D'}  \h \ee^{i\phi'} X_{\vct u''\oplus \vct w'} Z_{\vct v''\oplus \vct w''}  P_{\vct b''}Z_{\mat B''}X_{\mat A''}},
\end{equation*}
and $X_{ij}C$ is in  \intnf.
\medskip

\texttt{Case 4.1.2} : $(d_i,d_j)=(0,1)$. In this case, 

$Z_{\{i,j\}}^{\h} P_{\vct{d}}Z_{\{i,j\}}^{\h}=Z_{\{i,j\}}^{\h} P_jZ_{\{i,j\}}^{\h} P_{\vct d\oplus \vct e_j}
\stackrel{\eqref{conj-pj-zijh}}{=} P_i^{\h}X_{[ij]} P_j P_{\vct d\oplus \vct e_j}= P_i^{\h}X_{[ij]} P_{\vct d}$, hence
\begin{align*}
X_{[ij]}C&= H_{\vct a}\cdot \blue{P_i^{\h}X_{[ij]} P_{\vct d}} \cdot \red{Z_{\mat D'}  \h
           \ee^{i\phi'} X_{\vct u''\oplus \vct w'} Z_{\vct v''\oplus \vct w''} P_{\vct b''}Z_{\mat B''}X_{\mat A''}}\\
         &=H_{\vct a}P_i^{\h}\cdot \blue{X_{[ij]}}\red{P_{\vct d}Z_{\mat D'}  \h \ee^{i\phi'} X_{\vct u''\oplus \vct w'} Z_{\vct v''\oplus \vct w''} P_{\vct b''}Z_{\mat B''}X_{\mat A''}}.
\end{align*}

We can merge $X_{[ij]}$ in the red part using the same computation as in Case 1, starting from Equality \ref{case-31}) where $\vct a=\vct 0$. We obtain
\begin{equation}
  X_{[ij]}C= H_{\vct a}\cdot \blue{P_i^{\h}}\red{F_1},\label{end-case414}
\end{equation}
where $F_1$ is an \intnf\ such that $\vct a = \vct 0$, because no Hadamard gate is created in Case 1.
So, in order to merge $ P_i^{\h}$ with $F_1$, we can use the same computation as in Case 2 of the proof of Lemma \ref{stab-P}, starting from Equality \ref{case-22}.
We obtain
\begin{equation*}
X_{[ij]}C= \blue{H_{\vct a}}\red{F_2},
\end{equation*}
where $F_2$ is an  \intnf. Finally, we merge $H_{\vct a}$ with $F_2$ using Lemma \ref{stab-H},
and obtain thereby a rewriting in \intnf\ for $X_{[ij]}C$. 

\medskip

\texttt{Case 4.1.3} : $(d_i,d_j)=(1,0)$. We proceed as in case 4.1.2, swapping $i$ and $j$.
\medskip

\texttt{Case 4.1.4} : $(d_i,d_j)=(1,1)$. In this case,

$Z_{\{i,j\}}^{\h} P_{\vct{d}}Z_{\{i,j\}}^{\h}=Z_{\{i,j\}}^{\h} P_j P_iZ_{\{i,j\}}^{\h} P_{\vct d\oplus \vct e_i\oplus \vct e_j}
\stackrel{\eqref{conj-pj-zijh}}{=} P_i^{\h}X_{[ij]} P_j  P_j^{\h}X_{[ji]} P_i P_{\vct d\oplus \vct e_i\oplus \vct e_j}$.

Since Identity \eqref{conj-Pj} holds, $ P_j^{\h}X_{[ji]}$ and $ P_i P_{\vct d\oplus \vct e_i\oplus \vct e_j}$ commutes, so

$Z_{\{i,j\}}^{\h} P_{\vct{d}}Z_{\{i,j\}}^{\h}= P_i^{\h}X_{[ij]} P_{\vct{d}} P_j^{\h}X_{[ji]}$. Hence
\begin{align*}
X_{[ij]}C&= H_{\vct a}\cdot  \blue{P_i^{\h}X_{[ij]} P_{\vct{d}} P_j^{\h}X_{[ji]}}\cdot \red{Z_{D'}  \h
           \ee^{i\phi'} X_{u''\oplus w'} Z_{v''\oplus w''} P_{b''}Z_{B''}X_{A''}}\\
         &= H_{\vct a}P_i^{\h}X_{[ij]} P_{\vct{d}} P_j^{\h}\cdot \blue{X_{[ji]}}\red{Z_{D'}  \h
\ee^{i\phi'} X_{u''\oplus w'} Z_{v''\oplus w''} P_{b''}Z_{B''}X_{A''}}
\end{align*}
We can merge $X_{[ij]}$ in the red part using Case 1, starting form Equality \ref{case-31} in the special case where $\vct a=\vct d=\vct 0$. We obtain
\begin{equation*}
  X_{[ij]}C= H_{\vct a}P_i^{\h}X_{[ij]} P_{\vct{d}}\cdot \blue{P_j^{\h}}\red{F_1},
\end{equation*}
where $F_1$ is an \intnf\ such that $\vct a = \vct d=\vct 0$.
So, in order to merge $P_j^{\h}$ with $F_1$, we can use the same computation as in Case 2 and Case 2.1 of the proof of Lemma \ref{stab-P}, starting from Equality \ref{case-22} where $\vct d =\vct 0$. We obtain thereby an \intnf\ $F_2$ such that $\vct a=\vct 0$ because no Hadamard gate is created in these cases :
\begin{equation*}
  X_{[ij]}C= H_{\vct a}P_i^{\h}X_{[ij]}\cdot \blue{P_{\vct{d}}} \red{F_2}.
\end{equation*}
As $\vct a=\vct 0$ in the \intnf\ $F_2$, we can use Case 1 of the proof of Lemma \ref{stab-P} to merge $P_{\vct{d}}$ with $F_2$. We obtain
\begin{equation*}
  X_{[ij]}C= H_{\vct a}P_i^{\h}\cdot \blue{X_{[ij]}} \red{F_3},
\end{equation*}
where $F_3$ is an \intnf\ such that $\vct a= \vct 0$.  So, we can use again Case 1 to merge  $X_{[ij]}$ with $F_3$ and obtain
\begin{equation*}
  X_{[ij]}C= H_{\vct a}\cdot \blue{P_i^{\h}}\red{F_4},
\end{equation*}
where $F_4$ is an \intnf\ such that $\vct a= \vct 0$.
We note that we are in the same situation as in Case 4.1.2, Equality \ref{end-case414}. So  we obtain an \intnf\ for $X_{[ij]}C$ by proceeding in the same way.


\medskip

\texttt{Case 4.2} : $d_{ij}=1$. Let $\mat D'=\mat D\oplus \{\{i,j\}\}$, then $d'_{ij}=0$, and
\begin{align*}
  X_{[ij]}C&= H_{\vct a}\cdot \blue{Z_{\{i,j\}}^{\h}Z_{\{i,j\}}}\cdot \red{P_{\vct{d}}Z_{\mat D'} \h\ee^{i\phi} X_{\vct u} Z_{\vct v} P_{\vct{b}}Z_{\mat{B}}X_{\mat A}}\\
           &\stackrel{\ref{conj-zij-zijh}}{=} H_{\vct a}\cdot  \blue{H_i H_jX_{(ij)}Z_{\{i,j\}}^{\h}}  \cdot\red{ P_{\vct{d}}Z_{\mat D'} \h\ee^{i\phi} X_{\vct u} Z_{\vct v} P_{\vct{b}}Z_{\mat{B}}X_{\mat A}}.
\end{align*}

We use the conjugations rules \eqref{conj-Za-sij}, \eqref{conj-Pb-sij} and \eqref{conj-ZB-sij} by the $\swap$ gate $X_{(ij)}$ and we merge the Hadamard gates to obtain
\begin{equation*}
  X_{[ij]}C= H_{\vct a\oplus \vct e_i\oplus \vct e_j}\blue{Z_{\{i,j\}}^{\h}}  \red{P_{(ij)\vct d}Z_{(ij)\mat D'(ij)}
  \h\ee^{i\phi} X_{(ij)\vct u} Z_{(ij)\vct v} P_{(ij)\vct b}Z_{(ij)\mat B(ij)}X_{(ij)\mat A}}.
  \end{equation*}

Let $\mat D''= (ij)\mat D'(ij)$. Since $d'_{ij}=0$, then $d''_{ij}=0$, so we can proceed as in Case 4.1 and obtain thereby an \intnf\ for $X_{[ij]}C$.
  \end{proof}
    
  \subsection{Computing the generalized PZX form}

  We show that a \gpzx of a stabilizer circuit $C$ can be obtained in polynomial time by applying  an algorithm summarized in Figure \ref{C-to-NF}. This algorithm is called the $\ctogpzx$ algorithm.

  \begin{theo}[\textbf{The \gpzx for stabilizer circuits}]\label{nf-th}
  
    Any  $n$-qubit stabilizer circuit $C$ given as a product of $\ell$ Clifford gates, \textit{i.e.} $C=\prod_{k=1}^{\ell} M_k $, where
    $M_k\in \{ P_i, H_i,X_{[ij]}\mid 0\leq i,j\leq n-1\}$, can be written in polynomial time $O(\ell n^2)$ in the form
     \begin{equation}
      \ee^{\ii\phi}H_{\mathbf{r}} Z_{\mathbf{u}}P_{\mathbf{d}}Z_{\mathbf{D}} H_{\mathbf{s}} Z_{\mathbf{v}}P_{\mathbf b}Z_{\mathbf{B}}X_{\mathbf{A}}\ ,
    \end{equation}
     where     $\vct r, \vct u,\vct d, \vct s, \vct v, \vct b$ are vectors in $\F^n$, $\mat D$ and $\mat B$ are matrices in $\BG$, $\mat A$ is an invertible matrix in $\GL$, and $\phi\in\{k\frac{\pi}{4},k\in\Z\}$.
    \end{theo}

    \begin{proof} The computation of the \gpzx is divided into three steps A, B, C. We describe each step and  evaluate the number of operations needed to perform the step.\smallskip
      
      \textbf{Step A } : We prove by induction on the length $\ell$ of the input circuit that any stabilizer circuit  $C=\prod_{k=1}^{\ell} M_k $, where
      $M_k\in\{ P_i, H_i,X_{[ij]}\mid 0\leq i,j\leq n-1\}$, can be written in the \intnf\ $H_{\vct a} P_{\vct{d}}Z_{\mat{D}} \h\ee^{\ii\phi} X_{\vct u} Z_{\vct v} P_{\vct{b}}Z_{\mat{B}}X_{\mat{A}}$. The base case of the induction is clear : if $\ell=0$, then $C= I=H_{\vct a}\h$, where $\vct a$ is the vector of $\F^n$ with all entries equal to 1.
      To perform an induction step, we must prove that, for any $M\in\{ P_i, H_i,X_{[ij]}\mid 0\leq i,j\leq n-1\}$ and any Clifford circuit $C$ in  \intnf, the product $MC$ can be written in  \intnf. Clearly, the induction step is achieved by using Lemmas \ref{stab-H}, \ref{stab-P} and \ref{stab-X}.
      \smallskip
      
      Taking in account the different cases in the proofs of these three Lemmas, it appears that the algorithmic cost of merging a Clifford gate into the \intnf\ is bounded by the cost of the $\ctopzx$ algorithm.
      The complexity of the $\ctopzx$ algorithm  is $O(\ell' n)$, where $\ell'$ is the number of gate in the input circuit (see Figure \ref{C-to-PZX}). At each induction step, we apply this algorithm to subcircuits composed of $O(n)$ gates, so the cost of using $\ctopzx$ is $O(n^2)$ operations for each step.
      Starting from a stabilizer  circuit $C$ of length $\ell$, one needs $\ell$ induction steps to write $C$ in  \intnf, so we see that the number of operations needed in Step A is $O(\ell n^2)$.
\smallskip
      
      \textbf{Step B} : We write $C$ in the form $C= \ee^{\ii\phi}H_{\vct a} P_{\vct{d}}Z_{\mat{D}}\cdot X_{\vct u}^{\h}\cdot \h Z_{\vct v} P_{\vct{b}}Z_{\mat{B}}X_{\mat{A}}$
       and we use Identity \eqref{conj-xz-h} as well as the commutativity of the unitaries $P_{\vct{d}}$, $Z_{\mat{D}}$ and
       $Z_{\vct{u}}$ to obtain $C= \ee^{\ii\phi}H_{\vct a}Z_{\vct{u}}P_{\vct{d}}Z_{\mat{D}} \h Z_{\vct v} P_{\vct{b}}Z_{\mat{B}}X_{\mat{A}}$.
       The cost of this step is $O(1)$ (we neglect operations such as  initializing or copying which depend on the implementation of the $\ctogpzx$ algorithm).
       \smallskip

      \textbf{Step C } : The unitaries $H_{\vct r}$ and $H_{\vct s}$ appear after a straightforward simplification of the Hadamard gates. One defines the vectors  $\vct r$ and $\vct s$ as follows. Let $\Gamma$ be the set of qubits involved in the subcircuit $Z_{\vct{u}}P_{\vct{d}}Z_{\mat{D}}$, \textit{i.e} $\Gamma=\{i\mid u_i=1\}\cup\{i\mid d_i=1\}\cup\{i\mid \exists j, d_{ij}=1\}$. If $a_i=1$ and $i\notin \Gamma$ then $r_i=s_i=0$,
      otherwise $r_i=a_i$ and $s_i=1$.
      After this last step, $C=\ee^{\ii\phi}H_{\mathbf{r}} Z_{\mathbf{u}}P_{\mathbf{d}}Z_{\mathbf{D}} H_{\mathbf{s}} Z_{\mathbf{v}}P_{\mathbf b}Z_{\mathbf{B}}X_{\mathbf{A}}$,
      so $C$ is in the desired form. The cost of this simplification is $O(n^2)$.      
\end{proof}

\begin{figure}[h]
  \smallskip
  
  $\mathtt{ALGORITHM\ :}$  Compute a generalized \pzx for a stabilizer circuit.\smallskip

  $\mathtt{INPUT\ :}$  $C$, a stabilizer circuit given as a product of Clifford gates.\smallskip

  $\mathtt{OUTPUT\ :}$ An equivalent circuit to $C$, written in the \gpzx 
  
$\phantom{\mathtt{OUTPUT\ :}}$ $\ee^{\ii\phi}H_{\mathbf{r}} Z_{\mathbf{u}}P_{\mathbf{d}}Z_{\mathbf{D}} H_{\mathbf{s}} Z_{\mathbf{v}}P_{\mathbf b}Z_{\mathbf{B}}X_{\mathbf{A}}\ ,$
\smallskip

$\mathtt{Step\ A\ :}$ Write $C$ in  \intnf.

$\phantom{\mathtt{Step\ A\ :}}$ $C=H_{\vct a} P_{\vct{d}}Z_{\mat{D}} \h\ee^{\ii\phi} X_{\vct u} Z_{\vct v} P_{\vct{b}}Z_{\mat{B}}X_{\mat{A}}$
\smallskip

$\mathtt{Step\ B\ :}$ Move the \PauliX gates to the right.

$\phantom{\mathtt{Step\ B\ :}}$ $C=\ee^{\ii\phi}H_{\vct a}  Z_{\vct u}P_{\vct{d}}Z_{\mat{D}} \h Z_{\vct v} P_{\vct{b}}Z_{\mat{B}}X_{\mat{A}}$
\smallskip

$\mathtt{Step\ C\ :}$ Simplify the Hadamard gates.

$\phantom{\mathtt{Step\ C\ :}}$ $C=\ee^{\ii\phi}H_{\vct r}  Z_{\vct u}P_{\vct{d}}Z_{\mat{D}} H_{\vct s} Z_{\vct v} P_{\vct{b}}Z_{\mat{B}}X_{\mat{A}}$
\smallskip

$\mathtt{Return\ }$ $\ee^{\ii\phi}H_{\vct r}  Z_{\vct u}P_{\vct{d}}Z_{\mat{D}} H_{\vct s} Z_{\vct v} P_{\vct{b}}Z_{\mat{B}}X_{\mat{A}}$
\smallskip

\caption{ Algorithm  $\ctogpzx$ \label{C-to-NF}}
  \end{figure}

\begin{rem}
  The space complexity of the $\ctogpzx$ algorithm is only $O(n^2)$ (the space needed to store the matrices) if we merge each Clifford gate $M_k$ on-the-fly. We proceeded in this way to implement the $\ctogpzx$ algorithm as a Linux command ``$\mathtt{./stabnf}$'' with a text-based user interface.
  The source code of the command is available at

  \href{https://github.com/marcbataille/stabilizer-circuits-normal-forms/tree/graph_states}{$\mathtt{https://github.com/marcbataille/stabilizer\text{-}circuits\text{-}normal\text{-}forms}$}.
  
  The \textit{manual} mode of the command reproduces the induction steps described in the proofs of Lemmas \ref{stab-H}, \ref{stab-P} and \ref{stab-X}. In this mode, the user can write in the \gpzx a stabilizer circuit of arbitrary length and can observe the merging of each new Clifford gate into the normal form. 
  \end{rem}

  \begin{rem}
    The global phase $\phi$ of a quantum circuit is generally considered as irrelevant because it is physically unobservable. However, we decided  not to neglect $\phi$ during the computation process of the normal form, because knowing its exact value is, at least, of mathematical interest (for instance the group
    $\{ \ee^{\ii\phi}I | \phi\in\{k\frac{\pi}{4},k\in\Z\}\}$ has order 8 and this is related to the order of the group generated by the Clifford gates, see formula in the discussion  below). Besides, calculating the exact value of $\phi$ does not require much additional work. 
  \end{rem}

  \begin{rem}
    The $\ctogpzx$ algorithm can also take $\cz$, $\swap$, $\ZZ$, $\XX$, or $\YY$ gates as input since $Z_{\{i,j\}}= H_iX_{[ij]} H_i$,
    $S_{ij}=X_{(ij)}=X_{[ij]}X_{[ji]}X_{[ij]}$, $ Z_i= P_i^2$, $ X_i= H_i P_i^2 H_i$ and $ Y_i= P_i X_i P_i^{-1}= P_i H_i P_i^2 H_i P_i^3$.
    All theses gates are accepted as input of the  command ``$\mathtt{./stabnf}$'' in manual mode.
    \end{rem}

\begin{rem}
  The $\ctogpzx$ algorithm can be applied to an input circuit consisting only of $\PP, \cz$ and $\cnot$ gates. In this case, the output circuit is in the \pzx  $ Z_{\vct{v}} P_{\vct{b}}Z_{\mat{B}}X_{\mat{A}}$. So the $\ctogpzx$ algorithm is an extension of the $\ctopzx$ algorithm to any stabilizer circuit.
\end{rem}


\subsection{Discussion\label{discussion}} 
  
We discuss some questions related to the implementation of the \gpzx as a quantum circuit and we compare the \gpzx  to other recent normal forms for stabilizer circuits.

In order to implement the unitary $X_{\mat{A}}$ of the \gpzx as a $\cnot$ subcircuit,  we need to write the matrix $\mat A$ as a product of transvections. To this end, one can apply an algorithm proposed in 2004 by Patel \textit{et al.} \cite{2004PMH}. This algorithm is superior to the classical Gaussian elimination as it allows a decomposition in $O(n^2/\log n)$ transvections \cite[Theorem 1]{2004PMH} whereas the number of transvections in the decomposition obtained by the Gauss-Jordan algorithm is bounded by $n^2$ \cite[Proposition 10]{2020B}. In the rest of this paper, we refer to the Patel \textit{et al.}'s algorithm with parameter $m$ equal to $\lceil\log_2(n)/2\rceil$ as the $\AtoX$ algorithm  (see \cite{2004PMH} for the definition of $m$). It is important to remark that the $\AtoX$ algorithm does not return, in general, an optimized decomposition in transvections of the  matrix $\mat A$. 
Of course, it is possible to optimize a $\cnot$ circuit by a brute force algorithm but the cost is exponential and the algorithm can be used in practice only for small values of $n$. The method is as follows. First build
the Cayley graph of the group $\GL$ by Breadth-first search, then find in this
graph the matrix $\mat A$ corresponding to that $\cnot$ circuit. We implemented  this algorithm in the C language and the  
source file \texttt{cnot\_opt.c} is available at

\href{https://github.com/marcbataille/cnot-circuits/tree/master/optimization}{$\mathtt{https://github.com/marcbataille/cnot\text{-}circuits/tree/master/optimization}$}. Run on a basic laptop, the program allows to optimize in a few seconds any $\cnot$ circuit up to 5 qubits.\medskip

  Note that the unitary operators $Z_{\vct{u}}P_{\vct{d}}$ and $Z_{\vct{v}}P_{\vct{b}}$ in the \gpzx can be implemented as subcircuits of phase gates since $\ZZ=\PP^2$. So the unitary operator described by the \gpzx can be implemented as a quantum circuit of type
  \begin{equation}
   CX-CZ-P-H-CZ-P-H,\label{nf-circ}
\end{equation}
where  $CX$ (resp. $CZ$) is a subcircuit of $\cnot$ (resp. $\cz$) gates,
$P$ (resp. $H$) is a subcircuit of Phase (resp. Hadamard) gates
  \footnote{The reader who is not used to quantum circuits must pay attention to the following fact: a circuit acts to the right of the ket $\ket{\psi}$
  presented to its left but the associated operator acts to the left of $\ket{\psi}$, so the order of the gates in the circuit \ref{nf-circ} is inverted comparing to the \gpzx \ref{nf-formula}.}. 
The \gpzx has some similarities with the normal forms proposed in 2020 by Duncan \textit{et al.} \cite[Section 6]{2020DKPV} (H-P-CZ-CX-H-CZ-P-H) or
by Bravyi and Maslov \cite[lemma 8]{2020BM} (X-Z-P-CX-CZ-H-CZ-H-P). These two forms have, like the form \ref{nf-circ}, exactly two $\cz$ layers and one $\cnot$ layer but the number of single qubit layers is different. The form  H-P-CZ-CX-H-CZ-P-H proposed in \cite{2020DKPV} contains three layers of Hadamard gates whereas the form \ref{nf-circ} contains only two layers of Hadamard gates. If we merge the Z layer with the P layer at the beginning of the form X-Z-P-CX-CZ-H-CZ-H-P \cite{2020BM} (as we did for the form \ref{nf-circ}), the resulting form contains five single qubit layers, whereas the form \ref{nf-circ} contains four single qubit layers. So the \nf proposed in this paper can be considered as a slight simplification of the two forms mentioned above since it contains one single qubit layer less.
\medskip

The Clifford Group (defined as the normalizer of the Pauli Group in the unitary group $\UG$) is infinite. However the group generated by the gate set $\{H_i,P_i,X_{[ij]}\mid 0\leq i,j\leq n-1\}$ is a finite subgroup of the Clifford Group and its order is
$8\times 2^{2n}|\mathrm{Sp}_{2n}(\F)|=2^{n^2+2n+3}\prod_{j=1}^n(4^j-1)$
, where $\mathrm{Sp}_{2n}(\F)$ is the symplectic group over $\F$ in dimension $2n$ (see \textit{e.g.} \cite{1998CRSS}).
In \cite{2018MR}, the authors consider that the  number of Boolean degrees of
freedom in that group is $2n^2+O(n)$, since its order is $2^{2n^2+O(n)}$.
They deduce thereby that a \nf for stabilizer circuits
must have at least $2n^2+O(n)$ degrees of freedom. As a $\cnot$ layer adds $n^2$ degrees of freedom, a $\cz$ layer adds $n(n-1)/2$ degree of freedom
and the single qubit layers add a linear amount of degree of freedom (this is a direct consequence of the group orders of $\czg$ and $\cnotg$, see \cite[Section I]{2018MR} and \cite[Section 6]{2020DKPV}), it is easy to see that all three normal forms mentioned in this discussion have $2n^2+O(n)$ degrees of freedom and are therefore asymptotically optimal in the sense defined in \cite{2018MR}.\smallskip

\section{Implementing stabilizer states and graph states\label{gs}}

\subsection{The connection between stabilizer states and graph states}

A stabilizer state $\ket{S}$ for a $n$-qubit register can be written in the form
\begin{equation}
  \ket{S}=C\ket{0}^{\otimes n}, \label{ketS}
\end{equation}
where $C$ is a product of Clifford gates \cite[Theorem 1]{2004AG}. A graph state $\ket{G}$ is a special case of a stabilizer state that can be written in the form
\begin{equation}
  \ket{G}=Z_{\mat{B}}\ket{+}^{\otimes n}=Z_{\mat{B}} \h\ket{0}^{\otimes n}, \label{ketG}
  \end{equation}
  where $\ket{+}= \HH\ket{0}=\frac{1}{\sqrt{2}}(\ket{0}+\ket{1})$ is the eigenvector corresponding to the eigenvalue 1 of the \PauliX gate, $Z_{\mat B}$ is a product of $\cz$ gates defined by a matrix $\mat B$ in $\BG$ and $\h=\HH^{\otimes n}$\cite{2006HD}. The graph $G$ associated to the graph state $\ket{G}$ is the graph of order $n$ whose vertices are labeled by the $n$ qubits and whose set of edges is $\{\{i,j\}\mid b_{ij}=1\}$.\medskip
  
Let $\ket{S}=C\ket{0}^{\otimes n}$ be a stabilizer state. Applying the $\ctogpzx$ algorithm up to Step B to the stabilizer circuit $C$ yields
$C=\ee^{\ii\phi} H_{\vct a} Z_{\vct u} P_{\vct{d}}Z_{\mat{D}} \h Z_{\vct v} P_{\vct{b}}Z_{\mat{B}}X_{\mat{A}}$. Since the unitary $Z_{\vct v}P_{\vct{b}}Z_{\mat{B}}X_{\mat{A}}$  has no effect on the ket $\ket{0}^{\otimes n}$, one has, neglecting the global phase $\phi$ :  $\ket{S}= H_{\vct a} Z_{\vct u} P_{\vct{d}}Z_{\mat{D}} \h\ket{0}^{\otimes n}$. Hence
  $\ket{S}= H_{\vct a} Z_{\vct u} P_{\vct{d}}\ket{G}$,
where $\ket{G}$ is the graph state $Z_{\mat{D}} \h\ket{0}^{\otimes n}$. So, using the $\ctogpzx$ algorithm, we obtain  a new proof of a theorem from Van den Nest \textit{et al.} \cite[theorem 1]{2004VDN} that asserts the equivalence under local Clifford operations of any stabilizer state $\ket{S}$ to a graph state $\ket{G}$ : there exists a stabilizer circuit $C'$ consisting only of local Clifford gates (\textit{i.e.} phase and Hadamard gates) and a graph state $\ket{G}$ such that  $\ket{S}=C'\ket{G}$. Moreover, the $\ctogpzx$ algorithm provides a possible construction of the circuit $C'$ and the graph $G$.

\begin{theo}[\textbf{Normal form of a stabilizer state}] \label{stab-graph}
  
  For any stabilizer state $\ket{S}$, there exists a graph state $\ket{G}$ and 3 vectors $\vct a, \vct u, \vct d$ in $\F^n$ such that
\begin{equation}
  \ket{S}= H_{\vct a} Z_{\vct{u}} P_{\vct d}\ket{G}.\label{stab-graph-values}
\end{equation}
\end{theo}
\medskip

Because of Theorem \ref{stab-graph}, implementing a stabilizer state is equivalent to implementing a graph state, up to a circuit of local Clifford gates.
So, in the rest of this section, we focus on the implementation of a graph state as a circuit in a quantum machine.

\subsection{Reducing the two-qubit gate count of a graph state}

We address the following question : what kind of pretreatment can be done in the classical circuit of $\cz$ and Hadamard gates that  implements a graph state $\ket{G}=Z_{\mat{B}}\ket{+}^{\otimes n}$ in order to reduce the two-qubit gate count ? We propose an implementation based on the gate set $\{\HH,\ZZ,\cz,\cnot\}$ obtained thanks to the use of Identity \eqref{conj-ZB-XA} ($X_{\mat{A}}Z_{\mat{B}}X_{\mat{A}}^{-1}= Z_{q_{\mat B}(\mat A^{-1})}Z_{\mat A^{-\T}\mat B \mat A^{-1}})$ together with the $\AtoX$ algorithm. The main idea is as follows.  The two-qubit gate count in a $\cz$ circuit is at most $n(n-1)/2$ gates, while the $\AtoX$ algorithm  allows an implementation  of $X_{\mat{A}}$ in $O(n^2/\log n)$ $\cnot$ gates.
Hence, if we find an equivalent circuit to the $\cz$ circuit corresponding to the unitary $Z_{\mat{B}}$, in which the gate count is dominated by the $\cnot$ gates, we can expect a possible reduction of the initial circuit.

\begin{defi}
We say that a matrix $\mat B\in\BG$ is \emph{reduced} when each column and each line of $\mat B$ contains at most one non-zero entry, \textit{i.e.} $Z_{\mat B}$ corresponds to a $\cz$ circuit of depth 1.
  \end{defi}

\begin{lem}\label{B-red}
  For any $\mat B$ $\in\BG$, there exists an upper triangular matrix $\mat A\in\GL$ and a reduced matrix $\mat B_{\text{red}}\in\BG$ such that
  $\mat B_{\text{red}}=\mat A^{\T}\mat B \mat A$.
\end{lem}

\begin{proof}
  The matrix $\mat B$ is the matrix of an alternating bilinear form with respect to the canonical basis $(\mat e_i)_{i=0\dots n-1}$ of $\F^n$.  The equality
  $\mat B_{\text{red}}=\mat A^{\T}\mat B \mat A$ is just the classical change of basis formula, where $\mat A$ is the matrix of the new basis. A possible construction of $\mat A$ and $\mat B_{\text{red}}$ is given by the algorithm $\BtoB$ in Figure \ref{B-red-algo}. We use basically Gaussian elimination (\textit{i.e.}
  multiplication by transvection matrices, \textit{cf.} Proposition \ref{tij-mult}) on columns and rows of matrix $\mat B$ to construct step by step the matrices $\mat A$ and $\mat B_{\text{red}}$ (see a complete example in Section \ref{complete-example}).
\end{proof}

\begin{figure}[h]
  $\mathtt{ALGORITHM\ :}$   Reduction of a matrix in $\BG$.
  
  $\mathtt{INPUT\ :}$  $\mat B$, a matrix in $\BG$.

  $\mathtt{OUTPUT\ :}$ $(\mat B',\mat A)$, where
  

    $\quad \mat B'\in\BG$ is a reduced matrix congruent to $\mat B$,

    $\quad \mat A\in\GL$ satisfies the congruence relation $\mat B'=\mat A^{\T}\mat B \mat A$.


    $\mathtt{1}\quad\ \mat B'\leftarrow \mat B;\ \mat A\leftarrow \mat I;$


    $\mathtt{2}\quad\ $/* $\mathrm{pivot}[j]=\mathtt{true}$, \emph{if $j$ has already been chosen as a pivot} */

    $\mathtt{3}\quad\ \mathtt{for}\ j= 0 \ \mathtt{to}\ n-1\ \mathtt{do}$ 

    $\mathtt{4}\quad\quad\ \mathrm{pivot}[j]\leftarrow \mathtt{false}$;

    $\mathtt{5}\quad\ \mathtt{for}\ j= 0 \ \mathtt{to}\ n-2\ \mathtt{do}$ 
  
    $\mathtt{6}\quad\quad\ \mathtt{if}\ \mathrm{pivot}[j]\ \mathtt{or}\ \mathrm{card}\{i\mid b'_{ij}=1\}=0\ \mathtt{then}$
    
    $\mathtt{7}\quad\quad\quad\ \mathtt{continue};$

    $\mathtt{8}\quad\quad\ $/*\emph{ choosing pivot} */
    
    $\mathtt{9}\quad\quad\ p\leftarrow \mathrm{min}\{ i \mid b'_{ij}=1\};$

    $\mathtt{10}\quad\quad \mathrm{pivot}[p]\leftarrow \mathtt{true};$

    $\mathtt{11}\quad\quad $/* \emph{Step a : eliminating the remaining 1's on column $j$ and line $j$} */

    $\mathtt{12}\quad\quad \mathtt{for}\ r=p + 1\ \mathtt{to}\ n-1\ \mathtt{do} $
    
    $\mathtt{13}\quad\quad\quad  \mathtt{if}\ b'_{rj}=1\ \mathtt{then} $
    
    $\mathtt{14}\quad\quad\quad\quad \mat B'\leftarrow [rp]\mat B'[pr];$

    $\mathtt{15}\quad\quad\quad\quad \mat A\leftarrow \mat A[pr];$



    $\mathtt{16}\quad\quad$/* \emph{Step b : eliminating the remaining 1's on line $p$ and column $p$} */
    
    $\mathtt{17}\quad\quad\mathtt{for}\ c=j + 1\ \mathtt{to}\ n-1\ \mathtt{do} $
    
    $\mathtt{18}\quad\quad\quad  \mathtt{if}\ b'_{pc}=1\ \mathtt{then} $

    $\mathtt{19}\quad\quad\quad\quad \mat B'\leftarrow [cj]\mat B'[jc];$

    $\mathtt{20}\quad\quad\quad\quad \mat A\leftarrow \mat A[jc];$



    $\mathtt{21}\quad \mathtt{return} (\mat B',\mat A);$

    
    \caption{ Algorithm $\BtoB$\label{B-red-algo}}
  \end{figure}

\begin{theo}\label{graph-state}
  Any graph state $\ket{G}$ can be written in the form
\begin{equation}
  \ket{G}= Z_{\vct{v}}X_{\mat{A}}Z_{\mat B_{\text{red}}}\ket{+}^{\otimes n},\label{GSred}
\end{equation}
where $\mat v\in\F^n$, $\mat A\in\GL$ is an upper triangular matrix and $\mat B_{\text{red}}$ is a reduced matrix in $\BG$.
\end{theo}

\begin{proof}
  Let $\ket{G}=Z_{\mat{B}}\h\ket{0}^{\otimes n}$ be a graph state. Using lemma \ref{B-red} we construct $\mat B_{\text{red}}$ and $\mat A$ such that
  $\mat B_{\text{red}}=\mat A^{\T}\mat B \mat A$. Using Identity \eqref{conj-ZB-XA}, one obtains $X_{\mat{A}}Z_{\mat B_{\text{red}}}X_{\mat{A}}^{-1}=Z_{q_{\mat B_{\text{red}}}(\mat A^{-1})}Z_{\mat A^{-\T}\mat B_{\text{red}}\mat A^{-1}}$.
  Hence $Z_{\mat{B}}=Z_{\vct v}X_{\mat{A}}Z_{\mat B_{\text{red}}}X_{\mat{A}}^{-1}$, where $\vct v=q_{\mat B_{\text{red}}}(\mat A^{-1})$ . So
  $\ket{G}= Z_{\vct{v}}X_{\mat{A}}Z_{\mat B_{\text{red}}}X_{\mat A^{-1}} \h\ket{0}^{\otimes n}$.

  Since Identity \eqref{conj-XA-h} holds, we obtain
  $\ket{G}= Z_{\vct{v}}X_{\mat{A}}Z_{\mat B_{\text{red}}} \h X_{\mat A^{\T}}\ket{0}^{\otimes n}$. As a $\cnot$ circuit has no effect on the ket $\ket{0}^{\otimes n}$,
  one has $\ket{G}= Z_{\vct{v}}X_{\mat{A}}Z_{\mat B_{\text{red}}} \h\ket{0}^{\otimes n}$.
  \end{proof}

  In Section \ref{complete-example}, we provide a detailed example of the use of Theorem \ref{graph-state}. Note that the $\cz$ subcircuit in the form \eqref{GSred} has depth 1 and consequently all the $\cz$ gates can be applied at the same time. This observation has a practical utility because the decoherence time remains currently  an important technical concern in the experimental quantum computers.

  The form \eqref{GSred} allows to implement a graph state $\ket{G}=Z_{\mat{B}}\h\ket{0}^{\otimes n}$ in $O\left(\frac{n^2}{\log(n)}\right)$ two-qubit gates by using the $\AtoX$ algorithm on the matrix $\mat A$. In this implementation, the two-qubit gate count is asymptotically better than the
  bound $n(n-1)/2$  resulting from a basic implementation of the $Z_{\mat{B}}$ operator. Whether or not this pretreatment  brings a real practical advantage depends, however, on the value of the constant $\mathrm{c}$ such that the two-qubit gate count remains lower than $\mathrm{c}\frac{n^2}{\log(n)}$. In Table \ref{stat}, we propose a few statistics in order to evaluate the usefulness of the form \eqref{GSred} in terms of reduction of the two-qubit gate count.
  This table was filled by using the command ``$\mathtt{./stabnf}$'' in statistics mode. The source code of the command is available at
  \href{https://github.com/marcbataille/stabilizer-circuits-normal-forms/tree/graph_states}{$\mathtt{https://github.com/marcbataille/stabilizer\text{-}circuits\text{-}normal\text{-}forms}$}. We tested different samples of 200 random graph states, up to 300 qubits. The results show a clear superiority of the form $\ket{G}= Z_{\vct{v}}X_{\mat{A}}Z_{\mat B_{\text{red}}}\ket{+}^{\otimes n}$ over the classical form $\ket{G}=Z_{\mat B}\ket{+}^{\otimes n}$ in most samples. More precisely, let us define the \textit{density} $d$ of a graph state as being the quotient $\frac{\ell}{n(n-1)/2}$, where $\ell$ is the number of edges and $n$ the number of qubits. We observe that our method is efficient in all cases if the density of the graph state is greater than $0.6$. For a graph state of small density ($d\leq 0.2$), other methods have to be developed. 

  \begin{table}[h]
    \begin{center}
    \begin{tabular}{|c|c|c|c|c|c|}\hline
      \diagbox{$n$}{$\ell$}&$0.2\times max$&$0.4\times max$&$0.6\times max$&$0.8\times max$&$max=n(n-1)/2$\\\hline
      $5$&0\%&0\%&1\%&21\%&20\%\\\hline
      $10$&0\%&0\%&20\%&41\%&33\%\\\hline
      $20$&0\%&0\%&31\%&49\%&54\%\\\hline
      $50$&0\%&12\%&41\%&56\%&62\%\\\hline
      $100$&0\%&23\%&48\%&61\%&72\%\\\hline
      $200$&0\%&31\%&54\%&66\%&74\%\\\hline
      $300$&0\%&37\%&58\%&68\%&79\%\\\hline
    \end{tabular}
    \end{center}
\caption{Gain obtained on the implementation of a graph state $\ket{G}$ by using the form \eqref{GSred} ($\ket{G}= Z_{\vct{v}}X_{\mat{A}}Z_{\mat B_{\text{red}}}\ket{+}^{\otimes n}$) instead of the form $\ket{G}=Z_{\mat B}\ket{+}^{\otimes n}$. In each cell $(n,\ell)$ of the table, we computed the average two-qubit gate gain on a sample of 200 random graph states, where each graph has $n$ vertices and $\ell$ edges. Let $\ell'$ be the number of two-qubit gates in the circuit implementing the form \eqref{GSred}. The gain is defined as the difference $\ell-\ell'$ if $\ell>\ell'$ and 0 otherwise. It is expressed in percentage of $\ell$. \label{stat}}
    \end{table}

  \subsection{A complete example\label{complete-example}}
  We detail a complete  example illustrating Theorem \ref{graph-state} and Algorithm $\BtoB$.
  
Let $n=7$ and $\ket{G}=Z_{03}Z_{05}Z_{12}Z_{13}Z_{16}Z_{24}Z_{25}Z_{34}Z_{56}\ket{+}^{\otimes 7}$.\smallskip

We have $\ket{G}=Z_{\mat{B}}\ket{+}^{\otimes 7}$, where $\mat B=\begin{bmatrix}
  0&0&0&1&0&1&0\\
  0&0&1&1&0&0&1\\
  0&1&0&0&1&1&0\\
  1&1&0&0&1&0&0\\
  0&0&1&1&0&0&0\\
  1&0&1&0&0&0&1\\
  0&1&0&0&0&1&0\\
\end{bmatrix}$.
\medskip

\textbf{Stage 1 :} computing matrices $\mat A$ and $\mat B_{\text{red}}$.\smallskip

We apply Algorithm $\BtoB$ to the matrix $\mat B$.

After initializing each entry of the table \emph{pivot} to false, we describe step by step the execution of the main loop (lines $\mathtt{5}$ to $\mathtt{20}$ in Figure \ref{B-red-algo}).\medskip

  \framebox{$j=0$}\quad
  Choosing pivot : $p\leftarrow 3$; $\mathrm{pivot}[3]\leftarrow \mathtt{true};$\smallskip

  Step a :
$[53]\mat B[35]=
\begin{bmatrix}
  \bf{0}&\bf0&\bf0&\bf1&\bf0&\bf0&\bf0\\
  \bf0&0&1&1&0&1&1\\
  \bf0&1&0&0&1&1&0\\
  \bf1&1&0&0&1&0&0\\
  \bf0&0&1&1&0&1&0\\
  \bf0&1&1&0&1&0&1\\
  \bf0&1&0&0&0&1&0\\
\end{bmatrix}$

Step b : 
$[40][10][53]\mat B[35][01][04]=
\begin{bmatrix}
  \bf0&\bf0&\bf0&\bf1&\bf0&\bf0&\bf0\\
  \bf0&0&1&\bf0&0&1&1\\
  \bf0&1&0&\bf0&1&1&0\\
  \bf1&\bf0&\bf0&\bf0&\bf0&\bf0&\bf0\\
  \bf0&0&1&\bf0&0&1&0\\
  \bf0&1&1&\bf0&1&0&1\\
  \bf0&1&0&\bf0&0&1&0\\
\end{bmatrix}$
\medskip

\framebox{$j=1$}\quad 
Choosing pivot : $p\leftarrow 2$; $\mathrm{pivot}[2]\leftarrow \mathtt{true}$;\smallskip

Step a :
$[62][52][40][10][53]\mat B[35][01][04][25][26]=
\begin{bmatrix}
 \bf0&\bf0&\bf0&\bf1&\bf0&\bf0&\bf0\\
  \bf0&\bf0&\bf1&\bf0&\bf0&\bf0&\bf0\\
  \bf0&\bf1&0&\bf0&1&1&0\\
  \bf1&\bf0&\bf0&\bf0&\bf0&\bf0&\bf0\\
  \bf0&\bf0&1&\bf0&0&0&1\\
  \bf0&\bf0&1&\bf0&0&0&0\\
  \bf0&\bf0&0&\bf0&1&0&0\\
\end{bmatrix}$

Step b : 
$[51][41][62][52][40][10][53]\mat B[35][01][04][25][26][14][15]=
\begin{bmatrix}
 \bf0&\bf0&\bf0&\bf1&\bf0&\bf0&\bf0\\
  \bf0&\bf0&\bf1&\bf0&\bf0&\bf0&\bf0\\
  \bf0&\bf1&\bf0&\bf0&\bf0&\bf0&\bf0\\
  \bf1&\bf0&\bf0&\bf0&\bf0&\bf0&\bf0\\
  \bf0&\bf0&\bf0&\bf0&0&0&1\\
  \bf0&\bf0&\bf0&\bf0&0&0&0\\
  \bf0&\bf0&\bf0&\bf0&1&0&0\\
\end{bmatrix}$
\medskip

\framebox{$j=2$}\quad $\mathrm{pivot}[2]=\mathtt{true}$, so $\mathtt{continue}$ 
\medskip

\framebox{$j=3$}\quad  $\mathrm{pivot}[3]=\mathtt{true}$, so  $\mathtt{continue}$ 
\medskip

\framebox{$j=4$}\quad Choosing pivot :  $p\leftarrow 6$; $\mathrm{pivot}[6]\leftarrow \mathtt{true}$;\smallskip

Step a : $\mat B'$ remains unchanged

Step b : $\mat B'$ remains unchanged
\medskip

\framebox{$j=5$}\quad  null column, so $\mathtt{continue}$
\medskip

$\mathtt{return} (\mat B',\mat A)$, where

$\mat B'=\mat B_{\text{red}}=
\begin{bmatrix}
 \bf0&\bf0&\bf0&\bf1&\bf0&\bf0&\bf0\\
  \bf0&\bf0&\bf1&\bf0&\bf0&\bf0&\bf0\\
  \bf0&\bf1&\bf0&\bf0&\bf0&\bf0&\bf0\\
  \bf1&\bf0&\bf0&\bf0&\bf0&\bf0&\bf0\\
  \bf0&\bf0&\bf0&\bf0&\bf0&\bf0&\bf1\\
  \bf0&\bf0&\bf0&\bf0&\bf0&\bf0&\bf0\\
  \bf0&\bf0&\bf0&\bf0&\bf1&\bf0&\bf0\\
\end{bmatrix}$,

and $\mat A=[35][01][04][25][26][14][15]=
\begin{bmatrix}
  1&1&0&0&0&1&0\\
  0&1&0&0&1&1&0\\
  0&0&1&0&0&1&1\\
  0&0&0&1&0&1&0\\
  0&0&0&0&1&0&0\\
  0&0&0&0&0&1&0\\
  0&0&0&0&0&0&1\\
\end{bmatrix}$.
\medskip

\textbf{Stage 2} : computing the Pauli part $Z_{\vct{v}}$, where $\vct v=q_{\mat B_{\text{red}}}(\mat A^{-1})$.

The quadratic form $q_{\mat B_{\text{red}}}$ is defined by :

$q_{\mat B_{\text{red}}}([x_0,x_1,x_2,x_3,x_4,x_5,x_6]^t)=x_0x_3\oplus x_1x_2\oplus x_4x_6$,

and $\mat A^{-1}=[15][14][26][25][04][01][35]=\begin{bmatrix}
  1&1&0&0&1&0&0\\
  0&1&0&0&1&1&0\\
  0&0&1&0&0&1&1\\
  0&0&0&1&0&1&0\\
  0&0&0&0&1&0&0\\
  0&0&0&0&0&1&0\\
  0&0&0&0&0&0&1\\
\end{bmatrix}$.

Hence $q_{\mat B_{\text{red}}}(\mat A^{-1})=[0,0,0,0,0,1,0]^t$, so $Z_{\vct{v}}=Z_5$.
\medskip

\textbf{Stage 3}: applying the $\AtoX$ algorithm to matrix $\mat A$.

This yields $\mat A=[35][25][26][14][01][15]$.
\medskip

\textbf{Stage 4} : conclusion.

We deduced that $\ket{G}=Z_5X_{35}X_{25}X_{26}X_{14}X_{01}X_{15}Z_{03}Z_{12}Z_{46}\ket{+}^{\otimes 7}$

\subsection{Implementation of  graph states in the IBM quantum computers}

We deal with the case of a concrete implementation of graph states in a real-life quantum machine. Our intention is to show the practical usefulness that can have a pretreatment of the circuit based on Theorem \ref{graph-state}, in terms of reduction of the \emph{native} gate count in the \emph{compiled} circuit.
We implemented in the publicly available 5-qubit ibmq\_belem device
(\url{https://quantum-computing.ibm.com/}) the complete graph state
\begin{equation}
  \ket{K_5}=Z_{01}Z_{02}Z_{03}Z_{04}Z_{12}Z_{13}Z_{14}Z_{23}Z_{24}Z_{34}\ket{+}^{\otimes 5}.\label{gfull}
\end{equation}
This graph-state is of particular interest because it is LC (Local Clifford) equivalent, and thus SLOCC equivalent, to the entangled state
$\ket{\mathtt{GHZ}}_5=\frac{1}{\sqrt 2}(\ket{00000}+\ket{11111})$ (see \cite{1990GHZ} for the first introduction of the $\ket{\mathtt{GHZ}}$ state and \cite[Section 4.1]{2006HD} for a proof of the equivalence). To write $\ket{K_5}$ in the form \eqref{GSred}, we simply use our command \texttt{./stabnf} in manual mode and obtain :
 \begin{equation}
   \ket{K_5}= Z_2 Z_3X_{34}X_{23}X_{12}X_{02}X_{24}X_{23}Z_{01}Z_{23}\ket{+}^{\otimes 5}.\label{gfull-red}
 \end{equation}
 Observe that the form \ref{gfull-red} contains only 8 two-qubit gates comparing to the 10 $\cz$ gates of the form \ref{gfull}, which is a substantial reduction of 20\%.
 But what about the reduction if we consider the circuit, consisting exclusively of native gates, that is actually implemented in the quantum computer ? Is it still significant ? In the IBM quantum devices, the $\cz$ gate
 is not native  and is simulated thanks to Identity \eqref{zijxij}. The Hadamard gate is implemented from the $R_z(\pi/2)$ and $\sqrt{\XX}$ gates, since
 $ \HH=\ee^{\ii\frac{\pi}{4}}R_z(\pi/2)\sqrt{\XX}R_z(\pi/2)$, where $\sqrt{\XX}=\dfrac12\begin{bmatrix}1+\ii&1-\ii\\1-\ii&1+\ii\end{bmatrix}$ and
   $R_z(\theta)=\begin{bmatrix}\ee^{-\ii\frac{\theta}{2}}&0\\
     0&\ee^{\ii\frac{\theta}{2}}
   \end{bmatrix}$. Moreover full connectivity is not achieved and the direct connections allowed between two qubits are given by a graph. The graph of the 5-qubit ibmq\_belem device is $\{\{1,0\},\{1,2\},\{1,3\},\{3,4\}\}$.
   So, to implement a $\cnot$ gate between qubits without direct connection (\textit{e.g.} qubits 2 and 3), it is necessary to simulate it from the native $\cnot$ gates using methods we described in a previous work \cite[Section 3]{2020B}. Due to its similarities to the compilation process in classical computing, the rewriting process that transforms an input circuit with measurements into a native gate circuit giving statistically the same measurement results, is called \textit{transpilation} on the IBM quantum computing website. \medskip

   The quantum circuits below were produced using the publicly available IBM Quantum Composer \href{https://quantum-computing.ibm.com/}{https://quantum-computing.ibm.com/}. First, we present the circuits (before and after transpilation) in the case of an implementation of $\ket{K_5}$
   corresponding to the form \ref{gfull}.  
  
   \bigskip
   
   
   $\mathtt{INPUT}:\ket{K_5}=Z_{\mat{B}}\ket{+}^{\otimes 5}=Z_{01}Z_{02}Z_{03}Z_{04}Z_{12}Z_{13}Z_{14}Z_{23}Z_{24}Z_{34}\ket{+}^{\otimes 5}$
   
   \includegraphics[scale=0.4, viewport=0cm 0cm 30.5cm 11cm, clip=true]{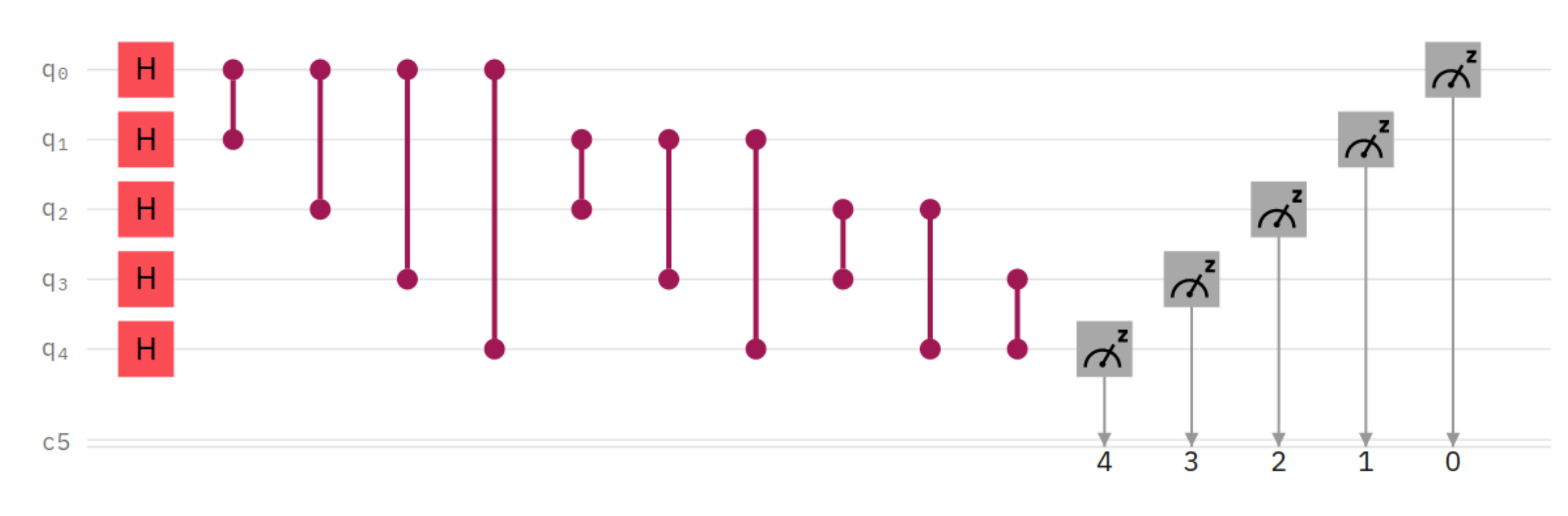}\raisebox{2.1cm}{$\qquad\xrightarrow{transpilation}$}
      
   $\mathtt{OUTPUT} : $
   
   \includegraphics[scale=0.4,viewport=0cm 0cm 33.9cm 11cm, clip=true]{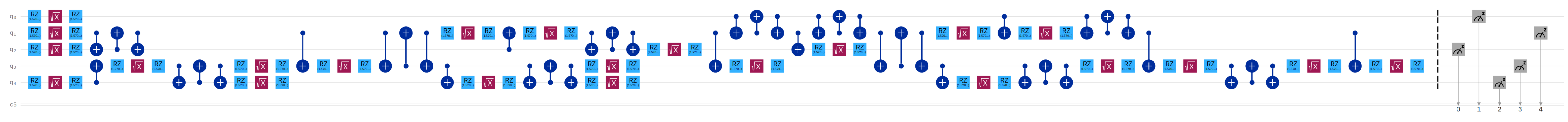}\raisebox{2.45cm}{$\quad\dots$}

   \raisebox{2.45cm}{$\dots\quad$}\includegraphics[scale=0.4,viewport=33.9cm 0cm 67.3cm 11cm, clip=true]{circuit_ghz5_full_trans}\raisebox{2.45cm}{$\quad\dots$}
   
   \raisebox{2.45cm}{$\dots\quad$}\includegraphics[scale=0.4,viewport=67.3cm 0cm 100.9cm 11cm, clip=true]{circuit_ghz5_full_trans}\raisebox{2.45cm}{$\quad\dots$}

   \raisebox{2.45cm}{$\dots\quad$}\includegraphics[scale=0.4,viewport=100.9cm 0cm 133cm 11cm, clip=true]{circuit_ghz5_full_trans}

   \medskip
   
   We remark that the transpiled circuit based on the form \ref{gfull} contains 43 $\cnot$ gates and 69 single qubit gates.
   
   Then, we show the circuits (before and after transpilation) implementing the same graph state $\ket{K_5}$ written in the form \ref{gfull-red}.\medskip

   $\mathtt{INPUT} : \ket{K_5}= Z_{\vct{v}}X_{\mat{A}}Z_{B_{\text{red}}}\ket{+}^{\otimes 5}= Z_2 Z_3X_{34}X_{23}X_{12}X_{02}X_{24}X_{23}Z_{01}Z_{23}\ket{+}^{\otimes 5}$

   \includegraphics[scale=0.4, viewport=0cm 0cm 29cm 11cm, clip=true]{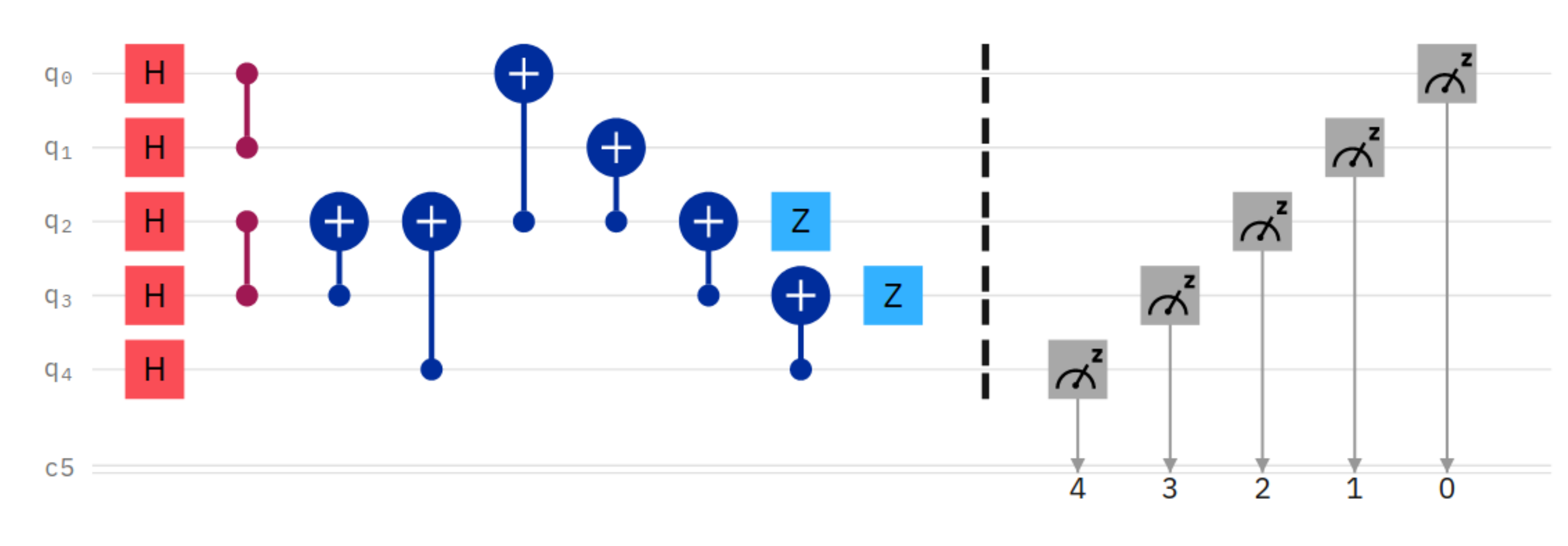}\raisebox{2.1 cm}{$\qquad\xrightarrow{transpilation}$}

   $\mathtt{OUTPUT} : $
   
   \includegraphics[scale=0.4, viewport=0cm 0cm 34cm 11cm, clip=true]{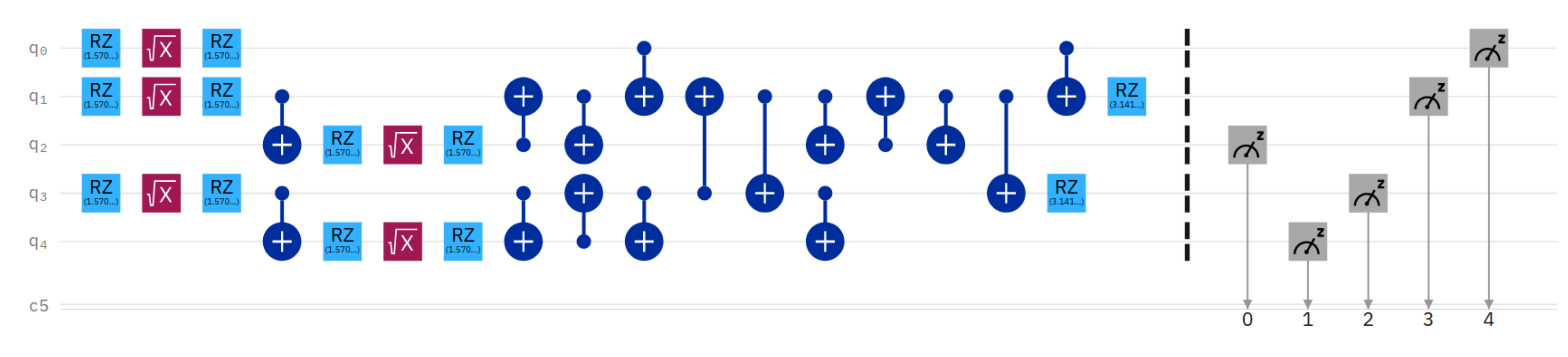}\raisebox{2.45cm}{$\quad\dots$}

   \includegraphics[scale=0.4, viewport=34cm 0cm 50cm 11cm, clip=true]{circuit_ghz5_full_czred_trans}
   \medskip
   
   In this case, the transpiled circuit contains only 16 $\cnot$ gates and  17 single gates. So, using our method, we obtain a reduction of 63\% of the two-qubit gate count, comparing to the naive implementation based on the form \ref{gfull-red}.

   \begin{table}[h]
     \begin{center}
       \footnotesize{
         \begin{tabular}{|c|c|c|c||c|c|}\hline
           \multicolumn{6}{|c|}{Implementation in the 5-qubit ibmq\_belem device}\\\hline
           \multirow{2}*{Ref.}&\multicolumn{3}{c||}{INPUT}&\multicolumn{2}{c|}{OUTPUT}\\ \cline{2-6}
                              &Circuit&Count&Gain&{Count}&Gain\\ \hline\hline
           1(a)&$Z_{01}Z_{02}Z_{03}Z_{04}Z_{12}Z_{13}Z_{14}Z_{23}Z_{24}Z_{34}\ket{+}^{\otimes 5}$&10&\multirow{2}*{20\%}&43&\multirow{2}*{63\%}\\ \cline{1-3} \cline{5-5}    
           1(b)&$Z_2 Z_3X_{34}X_{23}X_{12}X_{02}X_{24}X_{23}Z_{01}Z_{23}\ket{+}^{\otimes 5}$&8&&16&\\  \hline  \hline
           
           2(a)&$Z_{02}Z_{03}Z_{04}Z_{13}Z_{14}Z_{23}Z_{24}Z_{34}\ket{+}^{\otimes 5}$&8&\multirow{2}*{25\%}&26&\multirow{2}*{19\%}\\ \cline{1-3} \cline{5-5}    
           2(b)&$Z_3X_{34}X_{23}X_{14}X_{03}Z_{02}Z_{13}\ket{+}^{\otimes 5}$&6&&21&\\  \hline  \hline
           
           3(a)&$Z_{01}Z_{02}Z_{03}Z_{04}Z_{12}Z_{13}Z_{23}Z_{24}\ket{+}^{\otimes 5}$&8&\multirow{2}*{0\%}&35&\multirow{2}*{40\%}\\ \cline{1-3} \cline{5-5}    
           3(b)&$Z_2Z_3 X_{23} X_{24} X_{12} X_{02} X_{04} X_{23} Z_{01} Z_{23}\ket{+}^{\otimes 5}$&8&&21&\\  \hline  \hline
           
           4(a)&$Z_{01}Z_{02}Z_{04}Z_{12}Z_{13}Z_{23}Z_{34}\ket{+}^{\otimes 5}$&7&\multirow{2}*{14\%}&28&\multirow{2}*{32\%}\\ \cline{1-3} \cline{5-5}    
           4(b)&$Z_2 X_{12} X_{02} X_{14} X_{03} Z_{01} Z_{24} \ket{+}^{\otimes 5}$&6&&19&\\  \hline  \hline
           
           \multicolumn{6}{|c|}{Implementation in the 15-qubit ibmq\_melbourne device}\\\hline\hline
           5(a)&$Z_{02}Z_{03}Z_{04}Z_{13}Z_{14}Z_{15}Z_{23}Z_{25}Z_{26}Z_{34}Z_{35}Z_{45}Z_{46}Z_{56}\ket{+}^{\otimes 7}$&14&\multirow{2}*{21\%}&41&\multirow{2}*{22\%}\\ \cline{1-3} \cline{5-5}    
           5(b)&$Z_3 X_{34} X_{23} X_{35} X_{16} X_{06} X_{04} X_{03} Z_{02} Z_{13} Z_{46} \ket{+}^{\otimes 7}$&11&&32&\\  \hline  \hline
           
           6(a)&$Z_{03}Z_{05}Z_{12}Z_{13}Z_{16}Z_{24}Z_{25}Z_{34}Z_{56}\ket{+}^{\otimes 7}$&9&\multirow{2}*{0\%}&33&\multirow{2}*{18\%}\\ \cline{1-3} \cline{5-5}    
           6(b)&$ Z_5X_{35}X_{25}X_{26}X_{14}X_{01}X_{15}Z_{03}Z_{12}Z_{46}\ket{+}^{\otimes 7}$&9&&27&\\  \hline 
         \end{tabular}
       }
       \small{\caption{Implementation of graph states in two publicly available quantum computers. Gains obtained by form
         (b) : $Z_{\vct{v}}X_{\mat{A}}Z_{\mat B_{\text{red}}}\ket{+}^{\otimes n}$ over form (a) : $Z_{\mat{B}}\ket{+}^{\otimes n}$ (two-qubit gate count) before transpilation (INPUT) and after transpilation (OUTPUT).\label{IBMtable}}}
   \end{center}
      \end{table}

        In Table \ref{IBMtable}, we present the gains obtained, before and after transpilation, by using the form $Z_{\vct{v}}X_{\mat{A}}Z_{\mat B_{\text{red}}}\ket{+}^{\otimes n}$ of a few 5-qubit graph states (in the 5-qubit ibmq\_belem device) and 7-qubit graph states (in the 15-qubit ibmq\_melbourne device). Again, we observe a significant gain on the transpiled circuit. Moreover the gain after transpilation is often higher than the gain before transpilation.  Roughly, this is due to hardware reasons (graph of the qubit network, native gates) and to software reasons (how the compiler works) but this observation deserves certainly further analysis. Actually, a complete analysis should take in account the detailed technical specifications of the device as well as the source code of the compiler, which is beyond the scope of this paper. 
   Although our experiment is based on a few circuits implemented in some particular quantum computers, the results indicate that Theorem \ref{graph-state} can have practical useful applications, which was our initial purpose.

     \section{Conclusion and  future work}
Gottesman proved in his PhD thesis that any unitary matrix in the Clifford group is uniquely defined, up to a global phase, by its action by conjugation on the Pauli gates $ X_i$ and $ Z_i$ \cite[pp.41,42]{1997G}. This central statement of Gottesman stabilizer formalism can be used to compute normal forms for $n$-qubit stabilizer circuits via the symplectic group over $\F$ in dimension $2n$ (\textit{e.g.} \cite{2004AG,2018MR}).
In this paper we showed that it is possible to compute normal forms in polynomial time without using this formalism. We proposed a new method based on induction and on simple conjugation rules in the Clifford group.
The reader who is used to work with the symplectic group will notice that our induction process can also be applied inside this group, giving rise to a decomposition of type $\mat M_{\sigma}\begin{bmatrix}\mat I&\mat D\\ \mat 0&\mat I\end{bmatrix}  \begin{bmatrix}\mat I&\mat 0\\\mat B&\mat I\end{bmatrix}\begin{bmatrix}\mat A&\mat 0\\\mat 0&\mat A^{-\T}\end{bmatrix}$ for the symplectic matrix associated to the GenPZX form,
where $\mat B$ (resp. $\mat D$) is a symmetric matrix corresponding to $ P_{\vct{b}}Z_{\mat{B}}$ (resp. $ P_{\vct{d}}Z_{\mat{D}}$), $\mat A\in\GL$ is the invertible matrix corresponding to the $\cnot$ sub-circuit $X_{\mat{A}}$, and $\mat M_{\sigma}$ is a permutation matrix in dimension $2n$ corresponding to a circuit of Hadamard gates.
\medskip

In the NISQ era (Noisy Intermediate-Scale Quantum), noise in quantum gates strongly limits the reliability of quantum circuits and is currently a major technical concern. Developing optimization algorithms and heuristics to reduce the gate count in circuits is one of the solutions to improve reliability. In this article,  we proposed algorithms to reduce circuits implementing an important class of quantum states, namely the graph states, which are local Clifford equivalent to stabilizer states.
We realised a few experimental tests on quantum computers that highlight the utility of a pretreatment based on our algorithms to reduce the gate count in the compiled circuit implementing a graph state. We believe that the field of quantum circuits compilation is yet in its infancy but it will play an increasing significant role due to the current quick development of experimental quantum computers. We will continue to investigate reduction techniques related to the compilation of quantum circuits in future works.


\section{Acknowledgements}
The author acknowledges the use of the IBM Quantum Experience \href{https://quantum-computing.ibm.com/}{https://quantum-computing.ibm.com/}. The views
expressed are those of the author and do not reflect the official policy or position of IBM or
the IBM Quantum Experience team.

\bibliographystyle{plain}
\bibliography{biblio_MB}

\begin{thebibliography}{10}

\bibitem{2004AG}
Scott Aaronson and Daniel Gottesman.
\newblock Improved simulation of stabilizer circuits.
\newblock {\em Physical Review A}, 70(5), Nov 2004.

\bibitem{2020B}
Marc Bataille.
\newblock Quantum circuits of {CNOT} gates, 2020.
\newblock arXiv:2009.13247.

\bibitem{2019BL}
Marc Bataille and Jean-Gabriel Luque.
\newblock Quantum circuits of {cZ} and {SWAP} gates: optimization and
  entanglement.
\newblock {\em Journal of Physics A: Mathematical and Theoretical},
  52(32):325302, jul 2019.

\bibitem{2020BM}
Sergey Bravyi and Dmitri Maslov.
\newblock Hadamard-free circuits expose the structure of the clifford group,
  2020.
\newblock arXiv:2003.09412.

\bibitem{1998CRSS}
A.~R. {Calderbank}, E.~M. {Rains}, P.~M. {Shor}, and N.~J.~A. {Sloane}.
\newblock Quantum error correction via codes over {GF(4)}.
\newblock {\em IEEE Transactions on Information Theory}, 44(4):1369--1387,
  1998.

\bibitem{2020DKPV}
Ross Duncan, Aleks Kissinger, Simon Perdrix, and John van~de Wetering.
\newblock Graph-theoretic simplification of quantum circuits with the
  zx-calculus.
\newblock {\em Quantum}, 4:279, Jun 2020.

\bibitem{1997G}
Daniel Gottesman.
\newblock {\em Stabilizer Codes and Quantum Error Correction}.
\newblock PhD thesis, California Institute of Technology, Pasadena, CA, 1997.

\bibitem{1990GHZ}
Daniel~M. {Greenberger}, Michael~A. {Horne}, and Anton {Zeilinger}.
\newblock Bell's theorem without inequalities.
\newblock {\em American Journal of Physics}, 58 (12):1131, 1990.

\bibitem{2006HD}
M.~{Hein}, W.~{D{\"u}r}, J.~{Eisert}, R.~{Raussendorf}, M.~{Van den Nest}, and
  H.~J. {Briegel}.
\newblock {Entanglement in Graph States and its Applications}.
\newblock {\em arXiv e-prints}, pages quant--ph/0602096, February 2006.

\bibitem{2018MR}
Dmitri Maslov and Martin Roetteler.
\newblock Shorter stabilizer circuits via bruhat decomposition and quantum
  circuit transformations.
\newblock {\em IEEE Transactions on Information Theory}, 64(7):4729–4738, Jul
  2018.

\bibitem{2011NC}
Michael~A. Nielsen and Isaac~L. Chuang.
\newblock {\em Quantum Computation and Quantum Information: 10th Anniversary
  Edition}.
\newblock Cambridge University Press, New York, NY, USA, 10th edition, 2011.

\bibitem{2004PMH}
Ketan Patel, Igor Markov, and John Hayes.
\newblock Optimal synthesis of linear reversible circuits.
\newblock {\em Quantum Information and Computation}, 8, 05 2004.

\bibitem{2004VDN}
Maarten Van~den Nest, Jeroen Dehaene, and Bart De~Moor.
\newblock Graphical description of the action of local clifford transformations
  on graph states.
\newblock {\em Physical Review A}, 69(2), Feb 2004.

\end{thebibliography}

\end{document}